\documentclass{article}
\usepackage{graphicx}
\usepackage{amsfonts}
\usepackage{amsmath}
\usepackage{amssymb}
\usepackage{url}

\usepackage{dsfont}
\usepackage{fancyhdr}
\usepackage{indentfirst}
\usepackage{enumerate}
\usepackage[normalem]{ulem}
\usepackage{mathrsfs}
\usepackage{multirow}
\usepackage{diagbox}

 \usepackage[colorlinks=true,citecolor=blue]{hyperref}
\usepackage{amsthm}
\usepackage{color}

\definecolor{DarkGreen}{rgb}{0.2,0.6,0.2}

\definecolor{purple}{rgb}{0.6,0.3,0.8}

\usepackage{natbib}
\usepackage{comment}

\def\d{\mathrm{d}}

\newcommand{\R}{\mathbb{R}}

\newcommand{\p}{\mathbb{P}}
\newcommand{\X}{\mathcal{X}}

\newcommand{\M}{\mathcal{M}}

\newcommand{\id}{\mathds{1}}

\renewcommand{\ge}{\geqslant}
\renewcommand{\le}{\leqslant}
\renewcommand{\geq}{\geqslant}
\renewcommand{\leq}{\leqslant}
\renewcommand{\epsilon}{\varepsilon}

\renewcommand{\cdots}{\dots}
\theoremstyle{plain}
\newtheorem{theorem}{Theorem}
\newtheorem{corollary}{Corollary}
\newtheorem{lemma}{Lemma}
\newtheorem{proposition}{Proposition}
\theoremstyle{definition}
\newtheorem{definition}{Definition}
\newtheorem{example}{Example}

\theoremstyle{remark}

\newcommand{\thedate}{\today}

\usepackage{setspace}

\usepackage{footmisc}
\begin{document}
\title{Robust $\Lambda$-quantiles and {extremal distributions}}
\author{
Xia Han\thanks{School of Mathematical Sciences, LPMC and AAIS,  Nankai University, China. Email: xiahan@nankai.edu.cn}~~and~
Peng Liu\thanks{School of Mathematics, Statistics and Actuarial Science, University of Essex, UK. Email: peng.liu@essex.ac.uk} 
}

  \date{\thedate}

 \maketitle
\begin{abstract}

 In this paper, we investigate the robust models for $\Lambda$-quantiles with partial information regarding the loss distribution, where $\Lambda$-quantiles extend the classical quantiles by replacing the fixed probability level with a probability/loss function $\Lambda$. We find that, under some assumptions, the robust $\Lambda$-quantiles equal the $\Lambda$-quantiles of the {extremal distributions}. This finding allows us to obtain the robust $\Lambda$-quantiles by applying the results of robust quantiles in the literature. Our results are applied to  uncertainty sets characterized by three different constraints respectively:  moment constraints, probability distance constraints via the Wasserstein metric, and marginal constraints in risk aggregation.  We obtain some explicit expressions for robust $\Lambda$-quantiles by deriving the {extremal distributions} for each uncertainty set. These results are applied to optimal portfolio selection under model uncertainty.

\begin{bfseries}Key-words\end{bfseries}: Quantiles; $\Lambda$-quantiles; Distributionally robust optimization; Robust $\Lambda$-quantiles; {Extremal distributions}; Wasserstein distance; Portfolio selection.
\end{abstract}

 \section{Introduction}\label{sec:1}

To assess the impact of model misspecification and offer a robust quantification of the optimization problem, distributionally robust optimization (DRO) has been extensively  studied in recent years and has become a fast-growing field. In DRO, some object of interest is assessed over a set of alternative distributions, where the object can be some risk measures (functionals) or expected utility and the set of alternative distributions is called the uncertainty set or ambiguity set. 
Typically, we are concerned with the ``worst-case'' (maximal) and ``best-case'' (minimal) values of some risk measures over the
uncertainty set, called the robust risk measures, representing the robust quantification of the risk
measures, which is independent of the choice of the distributions lying in the uncertainty set.  The importance of the robustness of risk measures has been emphasized in the academic response of Basel Accord 3.5 in \cite{EPRWB14}, and the robustness of the risk measures has been studied in  e.g., \cite{CDS10} and \cite{HKP22}.

The uncertainty sets are always characterized by some constraints, representing the partially known information.
The moment constraints are popularly used to characterize the uncertainty sets in the DRO literature. This means that  some information on the moments of the distributions is known. Under this type of  uncertainty sets, robust-quantiles were studied in \cite{GOO03}, robust expected shortfall (ES) was considered in \cite{NSU10} and \cite{ZF09},  robust convex risk measures were investigated in \cite{L18}, and robust distortion risk measures and riskmetrics  were considered in \cite{PWW20},  \cite{SZ23}, \cite{SZ24}, and \cite{CLM23}.  One popular probability distance constraint is defined by the Wasserstein metric, where the Wasserstein distance represents the tolerance of the discrepancy of the alternative distributions with the baseline distribution;  see e.g., \cite{EK18} and \cite{BM19}.  The worst-case values of quantiles,  spectral risk measures and distortion riskmetrics over this uncertainty set were studied in \cite{LMWW22}, \cite{BPV24} and \cite{PWW20}. The marginal constraint is always imposed on the risk aggregation, where the marginal distributions are known but the dependence structure is completely unknown. For this uncertainty set, the robust-quantiles are the main concern in the literature, which is only solved for some special marginal distributions.  The robust-quantiles were derived for marginal distributions with monotone densities in \cite{WW16}, \cite{JHW16} and \cite{BLLW20}. We refer to \cite{PR12} and \cite{EKP20} for numerical studies. The uncertainty sets characterized by the above three constraints will be explored later. Recently, the robust models for quantiles under ambiguity were obtained in \cite{LMW24}, where the uncertainty set is a collection of probability measures, different from our setup in the paper.

{ In this paper, we focus on $\Lambda$-quantiles, first introduced by \cite{FMP14}, which extend the classical quantiles by replacing the fixed probability level with a probability/loss function. For $\Lambda: \R \to [0,1]$ and  a distribution function $F$,   one version of $\Lambda$-quantiles is defined as 
$$
q_{\Lambda}^-(F) = \inf \{x \in \mathbb R: F(x) \ge \Lambda(x)\}
$$
with the convention that $\inf\emptyset=\infty$. Four different definitions of $\Lambda$-quantiles introduced  in  \cite{BP22} will be given in Section \ref{sec:2}. 

The $\Lambda$-quantiles offer greater flexibility by choosing the function $\Lambda$  instead of fixing a confidence level, which allows one to handle multiple confidence levels simultaneously.  Typically, the $\Lambda$ function is chosen to be either increasing or decreasing,  and  the choice and estimation of the $\Lambda$ function were studied in \cite{HMP18}. Interestingly, $\Lambda$-quantiles  still retain some good properties of quantiles such as monotonicity, robustness, elicitability, consistency and locality; see  \cite{BPR17} and \cite{BP22}.  As shown in  \cite{HWWX23},  $\Lambda$-quantiles with decreasing $\Lambda$ functions satisfy quasi-star-shapedness, indicating that convex combination with a deterministic risk does not increase the risk,  which is a property weaker than quasi-convexity and penalizing the special type of concentration. In addition, $\Lambda$-quantiles with decreasing $\Lambda$ functions also satisfy cash-subadditivity, which is desirable for measuring the risk of the financial position in the presence of stochastic or ambiguous interest rates; see \cite{KR09} and \cite{HWWX23}. Recently, an axiomatization of $\Lambda$-quantiles was offered in \cite{BP22},  and  a new expression of $\Lambda$-quantiles with a decreasing $\Lambda$ function was provided in \cite{HWWX23}.}  
  
{  The application of $\Lambda$-quantiles has also attracted some attention recently. For instance,
  \cite{IPP22}  studied the risk contribution (sensitivity analysis) of $\Lambda$-quantiles; \cite{BBB23} and \cite{BCHW24} focused on the applications of $\Lambda$-quantiles  to optimal reinsurance design and premium calculation; \cite{L24} and \cite{XH24} investigated the risk-sharing problem with $\Lambda$-quantiles representing the agents' preferences, which was extended to the case with heterogeneous beliefs by \cite{LTW24}; and \cite{PV24} used the scoring function of $\Lambda$-quantiles for elicitability as the cost function to study  the  Monge-Kantorovich optimal transport.}  

 In this paper, we aim to  establish the robust models for $\Lambda$-quantiles with limited information on the underlying distribution of loss, where the function $\Lambda$  can  either decrease or increase. Despite the prevailing notion in existing literature that most researchers consider the risk in the worst-case scenario, wherein the decision maker (DM)  is extremely ambiguity-averse, empirical studies (see e.g., \cite{HT91} and \cite{KLT15}) demonstrate that the DM's attitude toward ambiguity is not uniformly negative. In fact, they may exhibit ambiguity-loving tendencies if they perceive themselves as knowledgeable or competent in the context.  Thus,  both the  worst and best possible values for $\Lambda$-quantiles  over a set of
alternative distributions  will be  considered.


 In Section \ref{sec:2}, we introduce four different definitions of $\Lambda$-quantiles given by \cite{BP22} and summarize some new properties of $\Lambda$-quantiles. In particular, we show the quasi-star-shapedness for all four $\Lambda$-quantiles with decreasing $\Lambda$ functions and show that all those four $\Lambda$-quantiles are cash-subadditive for decreasing $\Lambda$ functions  and cash-supadditive for increasing $\Lambda$ functions. Those properties motivate the study and application of $\Lambda$-quantiles with monotone $\Lambda$ functions. Moreover, we offer an alternative expression for one $\Lambda$-quantile with decreasing $\Lambda$ functions as an analog of the expression in \cite{HWWX23}, which enhances interpretability and offers  convenience for computing robust $\Lambda$-quantiles.

  In Section \ref{sec:3}, we show that, under some assumption, the robust $\Lambda$-quantiles over a general uncertainty set is equal to the $\Lambda$-quantile of the { extremal  distributions} over the same uncertainty set, which is the main finding of this paper.  { In this paper, the extremal distribution is an increasing function representing the pointwise supremum or infimum of the distribution functions over an uncertainty set.} This fact shows that the computation of robust $\Lambda$-quantiles can be decomposed into two steps under some conditions: i) compute the {extremal distributions}, and  ii) compute the $\Lambda$-quantiles of the obtained {extremal distributions}. This means that, under some conditions,  in order to compute the robust $\Lambda$-quantiles, it suffices to compute the {extremal distributions}, which can be derived via the results of robust quantiles.  The crucial assumption for our results is the attainability of {extremal distributions}. Verifying this property for popular uncertainty sets is not trivial, and will be discussed in Section \ref{sec:4}.  Moreover, we construct several counterexamples for the assumptions of our main results and discuss the existence of the optimal distributions in the robust $\Lambda$-quantiles. Notably,  our findings are closely {linked} to \cite{MWW22}, where a novel approach, called  the model aggregation approach, was introduced to evaluate the risk under model uncertainty for some risk measures including quantiles and ES. 

In Section \ref{sec:4}, we consider the uncertainty sets characterized by three different constraints respectively: moment constraints, probability constraints via the Wasserstein metric, and marginal constraints in risk aggregation. We obtain the {extremal distributions} for each uncertainty set replying on the results on robust quantiles over the same uncertainty sets in the literature and check the attainability of the {extremal distributions} for each uncertainty set. By applying our main results in Section \ref{sec:3}, we directly obtain the robust $\Lambda$-quantiles, showing the usefulness of our findings. Moreover, we apply our results to the portfolio selection problem with model uncertainty, where the risk of the portfolio is assessed by $\Lambda$-quantiles. In particular, for the uncertainty sets with marginal constraints in risk aggregation, we show that diversification is penalized for the portfolio of assets with the same marginal distribution, extending the results for quantiles in \cite{CLLW22}. Some numerical studies have been performed to show the values of the robust $\Lambda$-quantiles and optimal portfolio positions. The paper is concluded in Section \ref{sec:5}.


   \section{Notation and preliminaries}\label{sec:2}

Let $\mathcal X$ be a set of random variables {on} a given atomless probability space $(\Omega,\mathcal F,\p)$. 
  Denote by $\mathcal{M} $  the set of  cumulative distribution functions  (cdfs) of all random variables in $\X$.  { In this paper,  we consider the risk functional $\rho: \M\to\mathbb{R}$. For convenience,  we sometimes use $\rho(X)$ for $X\in \X$ to represent $\rho(F_X)$, where $F_X$ is the distribution of $X$.}  For a distribution $F\in\mathcal{M}$, its left quantile at level $\alpha\in [0,1]$ is defined as
$$q_{\alpha}^-(F)=F^{-1}(\alpha)=\inf\{x\in \R: F(x) \ge \alpha \},$$
and its right quantile  at level $\alpha\in [0,1]$ is given by $$q_{\alpha}^+(F)=F_+^{-1}(\alpha)=\inf\{x\in\R: F(x)>\alpha\}$$
with the convention that $\inf\emptyset =\infty$. Note that by the above definition, we have $q_0^-(F)=-\infty$ and $q_1^+(F)=\infty$.
We refer to  e.g., \cite{MFE15} for more discussions on the properties and applications of quantiles. 
Throughout this paper, terms such as ``increasing'' and ``decreasing'' are in the non-strict sense, and the notation $a\vee b$  (resp.~$a\wedge b$) is the maximum (resp.~minimum) between real numbers $a$ and $b$, and $a_+=a\vee 0$.

We next present four  definitions of $\Lambda$-quantiles given in  \cite{BP22}.  


 \begin{definition}\label{def:1}For $\Lambda:\R\to [0,1]$, and an increasing function $f:\R\to[0,1]$, the $\Lambda$-quantiles of $f$ are defined as\footnote{{Note that we use the notation $f$ here to distinguish it from $F$, as $f$ is not necessarily a distribution function.}}
 \begin{align}\label{lambdaL}
 q_\Lambda^-(f)=\inf\{x\in \R: f(x) \ge \Lambda(x) \},&~~q_\Lambda^+(f)=\inf\{x\in \R: f(x)>\Lambda(x) \},\nonumber\\
 \tilde{q}_\Lambda^-(f)=\sup\{x\in \R: f(x)<\Lambda(x) \},&~~\tilde{q}_\Lambda^+(f)=\sup\{x\in \R: f(x)\leq\Lambda(x) \},
  \end{align}
 where $\inf\emptyset=\infty$ and $\sup\emptyset=-\infty$ by convention.  \end{definition}
In Definition \ref{def:1}, agents can choose $\Lambda$ as either increasing or decreasing functions.  Without additional assumptions, the four $\Lambda$-quantiles can all differ; see Examples 4 and 5 of \cite{BP22}. In general, an increasing $\Lambda$ implies that the DM accepts only a smaller probability for larger losses. Conversely, a decreasing $\Lambda$ suggests a tendency to tolerate a greater probability as losses increase.

A simple example of  $\Lambda$-quantile is the two-level $\Lambda$-quantile presented in Example 7 of \cite{BP22}, where the $\Lambda$ function is defined as $\Lambda: x\mapsto \beta \id_{\{x< z\}} + \alpha \id_{\{x\geq z\}}$. In particular,  if $\Lambda$ is a constant, then  $\Lambda$-quantiles boil down to quantiles.

  A mapping $\rho: \X \rightarrow \R$ is said to satisfy  \emph{monotonicity} if $\rho(X) \leq \rho(Y)$ for all $X,Y\in \X$ with $X \leq Y$, and satisfy
 {\emph{cash-additivity}} if $\rho(X+m)=\rho(X)+m$ for all $X\in \X$ and $m\in\R$.   It is well known that $q^-_\alpha$ and $q^+_\alpha$  {satisfy  both properties}.  However, $\Lambda$-quantiles are not cash-additive in general.

 A mapping $\rho: \X \rightarrow \R$ is said to satisfy
{\emph{cash-subadditivity} (resp. cash-supadditivity)} if $\rho(X+m)\leq\rho(X)+m$ (resp. $\rho(X+m)\geq\rho(X)+m$) for all $X\in \X$ and $m\ge 0$. 
A mapping $\rho$ is said to be   {\emph{quasi-star-shaped}} if $\rho(\lambda X+(1-\lambda) t) \leqslant \max \{\rho(X), \rho(t)\}$  for all $X \in \mathcal{X}, t \in \mathbb{R}$ and $\lambda \in[0,1]$.
 In the context of capital requirements, {cash-subadditivity} or supadditivity  allows a non-linear increase in the capital requirement as cash is added to the future financial position.  Cash-subadditivity is desirable if the interest rate is stochastic or ambiguous, and it allows the current reserve and the future risk positions to preserve their own num\'{e}raire. Here cash-subadditivity reflects the time value of money; see \cite{KR09} and \cite{HWWX23}. Cash-subadditivity is  satisfied by convex loss-based risk measures proposed in \cite{CDH13}.  Quasi-star-shapedness means  that $\rho$ has quasi-convexity at each constant, representing some consideration of diversification, which  is weaker than quasi-convexity. Its theoretic-decision interpretation is available at  \cite{HWWX23}.

In Theorem 3.1 of \cite{HWWX23}, a representation of $q_{\Lambda}^{-}$ is given in terms of quantiles. In the following proposition,  we obtain a similar formula for $q_{\Lambda}^{+}$. For completeness, we also include the representation for $q_{\Lambda}^{-}$ in the following proposition.
\begin{proposition} If $\Lambda:\R\to [0,1]$ is a decreasing function, then  for $F\in \mathcal{M}$, we have
 \begin{equation}\label{eq:alt_pre}
q_{\Lambda}^{-} (F) = \inf_{x\in \mathbb{R}} \left\{q_{\Lambda(x)}^{-} (F) \vee x\right\}
= \sup_{x\in \mathbb{R}} \left\{q_{\Lambda(x)}^{-} (F) \wedge x\right\},
\end{equation}
and
 \begin{equation}\label{eq:alt_pre2}
q_{\Lambda}^{+}(F)
= \inf_{x\in \R} \left\{ q_{\Lambda(x)}^{+}(F)  \vee  x \right\}= \sup_{x\in \R} \left\{ q_{\Lambda(x)}^{+}(F)  \wedge  x \right\}.
\end{equation}
\end{proposition}
\begin{proof}
  Note that \eqref{eq:alt_pre} follows directly from Theorem 3.1 of \cite{HWWX23}.  Next, we show \eqref{eq:alt_pre2}.
Note that for $x\in \R$ and $t\in [0,1]$, $F(x) \le t $ implies that  $q^+_t(F)\ge x$.
Hence,  for a decreasing $\Lambda$, we have  \begin{align*}
q_{\Lambda}^{+}(F) =\tilde{q}_\Lambda^+(F)&= \sup \{x\in \R: F(x) \le \Lambda(x)\}
\\&  \leq \sup \{x\in \R: q_{\Lambda(x)}^{+}(F) \ge  x\}
\\&  = \sup \{q_{\Lambda(x)}^{+}(F) \wedge  x: q_{\Lambda(x)}^{+}(F) \ge  x\}
  \le \sup_{x\in \R} \left\{q_{\Lambda(x)}^{+}(F) \wedge  x \right\}.
\end{align*}
Moreover, because $F(x)>t $ implies  that  $q^+_t(F)\le x$, we have
\begin{align*}
q_{\Lambda}^{+}(F) &= \inf \{x\in \R: F(x) > \Lambda(x)\}
\\&  \geq \inf \{x\in \R: q_{\Lambda(x)}^{+}(F) \leq x\}
\\&  = \inf \{q_{\Lambda(x)}^{+}(F) \vee  x: q_{\Lambda(x)}^{+}(F) \leq x\}
  \ge \inf_{x\in \R} \left\{q_{\Lambda(x)}^{+}(F) \vee  x \right\}.
\end{align*}
Since  $q_{\Lambda(x)}^{+}(F) \wedge  x   \le q_{\Lambda(y)}^{+}(F) \vee  y $ for any $x,y\in \R$, we have $$
q_{\Lambda}^{+}(F)   \le \sup_{x\in \R} \left\{ q_{\Lambda(x)}^{+}(F)  \wedge  x \right\}
\le \inf_{x\in \R} \left\{ q_{\Lambda(x)}^{+}(F)  \vee  x \right\} \le q_{\Lambda}^{+}(F),
$$ which  implies \eqref{eq:alt_pre2}.
\end{proof}
In \cite{BP22}, many properties of $\Lambda$-quantiles  such as monotonicity, locality, and quasi-convexity on distributions are discussed.
We next summarize some of the properties possessed by $\Lambda$-quantiles that were not included in \cite{BP22} in  the following proposition.

\begin{proposition}\label{prop:1} For $F\in \mathcal{M}$, the following statements hold.
\begin{itemize}
\item[(i)]
     $q_{\Lambda}^{-}(F)\leq \tilde{q}_{\Lambda}^{-}(F)$ and $q_{\Lambda}^{+}(F)\leq \tilde{q}_{\Lambda}^{+}(F)$, and the inequalities become equalities if $\Lambda$ is decreasing.
     \item[(ii)] { All four $\Lambda$-quantiles $q_{\Lambda}^{-}, q_{\Lambda}^{+}, \tilde{q}_{\Lambda}^{-}$, and $\tilde{q}_{\Lambda}^{+}$ are  cash-subadditive if $\Lambda$ is decreasing and cash-supadditive if $\Lambda$ is increasing. Furthermore,  if $\Lambda:\R\to (0,1)$ is right-continuous and has at most finite discontinuous points, then the converse conclusion also holds.}
\item[(iii)] If $\Lambda$ is decreasing, all four $\Lambda$-quantiles $q_{\Lambda}^{-}, q_{\Lambda}^{+}, \tilde{q}_{\Lambda}^{-}$, and $\tilde{q}_{\Lambda}^{+}$ are  quasi-star-shaped.
\end{itemize}
\end{proposition}
\begin{proof} We first focus on (i). We first show $q_{\Lambda}^{-}(F)\leq \tilde{q}_{\Lambda}^{-}(F)$. Clearly, if $\tilde{q}_{\Lambda}^{-}(F)=-\infty$, then $F(x)\geq \Lambda(x)$ for all $x\in\R$. This implies that $q_{\Lambda}^{-}(F)=-\infty$. Hence, $q_{\Lambda}^{-}(F)\leq \tilde{q}_{\Lambda}^{-}(F)$ holds. If $\tilde{q}_{\Lambda}^{-}(F)=+\infty$, the inequality holds trivially. Next, we consider the case $x_0=\tilde{q}_{\Lambda}^{-}(F)\in\R$. By the definition, we have $F(x)\geq \Lambda(x)$ for all $x>x_0$, which implies that  $q_{\Lambda}^{-}(F)\leq x_0=\tilde{q}_{\Lambda}^{-}(F)$. We can similarly show $q_{\Lambda}^{+}(F)\leq \tilde{q}_{\Lambda}^{+}(F)$. The equalities for decreasing $\Lambda$ follow directly from Proposition 6 of \cite{BP22}.

Next, we show (ii).   Note that for $c\ge 0$,
  $q_\Lambda^-(X+c) = q^-_{\Lambda^c}(X)+c $, where $\Lambda^c(t)=\Lambda(t+c)$ for $t\in \R$. Moreover, $\Lambda^c\le \Lambda$ if $\Lambda$ is decreasing and $\Lambda^c\ge \Lambda$ if $\Lambda$ is increasing. By Proposition 3 of \cite{BP22}, we have $q_{\Lambda^c}^-(X)\leq q_\Lambda^-(X)$ if $\Lambda$ is decreasing, and $q_{\Lambda^c}^-(X)\geq q_\Lambda^-(X)$ if $\Lambda$ is increasing.
Therefore, if $\Lambda$ is decreasing,  $q_\Lambda^- (X+c) \le q_\Lambda^- (X)+c$, implying $q_\Lambda^-$ is cash-subadditive; if $\Lambda$ is increasing,  $q_\Lambda^- (X+c) \ge q_\Lambda^- (X)+c$, implying that $q_\Lambda^-$ is cash-supadditive. The analogous arguments can be employed to show the properties of other $\Lambda$-quantiles.
{ Next, we show the converse conclusion.  We denote all the discontinuous points of $\Lambda$ by $D$. Since $q_\Lambda^-$ is cash-subadditive, 
 we have, for $c\ge 0$,  $q_\Lambda^-(X+c) = q^-_{\Lambda^c}(X)+c \leq q_\Lambda^-(X)+c$, which implies that $q^-_{\Lambda^c}(X)\leq q_\Lambda^-(X) $  for any $c\ge 0$. Suppose that $F_X$ is continuous.  Let $x_0=q_\Lambda^-(X)$ and $x_1=q^-_{\Lambda^c}(X)$.  Then we have $x_1\leq x_0$. By the definition of $q_{\Lambda}^-$ and  the right-continuity of $\Lambda$, we have $F_X(x_1) \geq \Lambda(x_1+c),~ F_X(x)<\Lambda(x+c)$ for all $x<x_1$, and $F_X(x_0)\geq \Lambda(x_0),~ F(x)<\Lambda(x)$ for all $x<x_0$.  If $x_1<x_0$, then we have $\Lambda(x_1+c)\leq F_X(x_1)<\Lambda(x_1)$, which implies $\Lambda(x_1+c)\leq \Lambda(x_1)$. If $x_1=x_0$, then we have $F_X(x_1)\geq \max(\Lambda(x_1+c),\Lambda(x_1))$ and $F_X(x)<\min (\Lambda(x),\Lambda(x+c))$ for all $x<x_1$. If $x_1, x_1+c\notin D$, then it follows from the continuity of $\Lambda$ and $F_X$ that $\max(\Lambda(x_1+c),\Lambda(x_1))\leq F_X(x_1)\leq \min (\Lambda(x_1),\Lambda(x_1+c))$, which implies $\Lambda(x_1)=\Lambda(x_1+c)$. Hence, if $F_X$ is continuous and $x_1, x_1+c\notin D$, then we have $\Lambda(x_1+c)\leq \Lambda(x_1)$ with $x_1=q^-_{\Lambda^c}(X)$. For any $z\in\R$ and $\epsilon>0$, let $F_{z,\epsilon}(x)=\frac{(x-z+\epsilon)_+}{2\epsilon}\wedge 1,~x\in\R$. Then it follows that $q^-_{\Lambda^c}(F_{z,\epsilon})\in (z-\epsilon, z+\epsilon)$. Hence, fixing $c\geq 0$, $S_c:=\{q^-_{\Lambda^c}(X): F_X~\text{is continuous with bounded support}\}$ is a dense subset of  $\R$, which implies that $\Lambda(x+c)\leq \Lambda(x)$ for all $x\in S_c$ if $x, x+c\notin D$. It follows from the right-continuity of $\Lambda$ that $\Lambda(x+c)\leq \Lambda(x)$ for all $x\in \R$, which further implies that $\Lambda(x+c)\leq \Lambda(x)$ for all $x\in \R$ and $c\geq 0$. Therefore, we conclude that cash-subaditivity of $q_{\Lambda}^-$ implies that $\Lambda$ is decreasing.  

Next, we suppose that $q_{\Lambda}^-$ is cash-supadditive. Then for $c\geq 0$, we have $q_\Lambda^-(X+c) = q^-_{\Lambda^c}(X)+c \geq q_\Lambda^-(X)+c$, which implies that $q^-_{\Lambda^c}(X)\geq q_\Lambda^-(X) $  for any $c\ge 0$.  Using the same arguments as above, we can show that  $\Lambda(x+c)\geq \Lambda(x)$ for all $x\in \R$ and $c\geq 0$.
Using the similar arguments for $q_{\Lambda}^-$,  for other three $\Lambda$-quantiles, we can show that cash-subadditivity (cash-supadditivity) implies that $\Lambda$ is decreasing (increasing).}

Finally, we consider (iii). By (i), we have $\tilde{q}_{\Lambda}^{-}=q_{\Lambda}^{-}$. Hence,
it follows from  Theorem 3.1 of \cite{HWWX23} that  $\tilde{q}_{\Lambda}^{-}$ and $q_{\Lambda}^{-}$ are  quasi-star-shaped. Similarly, in light of (i), we have $\tilde{q}_{\Lambda}^{+}(F)=q_{\Lambda}^{+}(F)$. Thus it suffices to show that $q_{\Lambda}^{+}$ is quasi-star-shaped. We first show that $q_{\alpha}^+\vee x$ is quasi-star-shaped for $\alpha\in [0,1]$ and $x\in\R$.  If $\alpha=1$, it is trivial.  Then we suppose $\alpha\in [0,1)$. For $\lambda\in [0,1], t\in\R$ and $X\in\X$, we have  $$q_{\alpha}^{+}(\lambda X+(1-\lambda)t)  \vee  x=\left(\lambda q_{\alpha}^{+}(X)+(1-\lambda)t\right)  \vee  x\leq \max(q_{\alpha}^{+}(X)\vee x, t\vee x).$$
Hence, $q_{\alpha}^+\vee x$ is quasi-star-shaped. In light of Lemma 3.1 of \cite{HWWX23} and the expression \eqref{eq:alt_pre2}, we have that $q_{\Lambda}^{+}$ is quasi-star-shaped.
\end{proof}

 {Cash-supadditivity} is the dual property of  {cash-subadditivity}.  As we can see from (ii) of Proposition \ref{prop:1}, whether a $\Lambda$-quantile satisfies cash-subadditivity or cash-supadditivity depends on the direction of monotonicity of the $\Lambda$ function. The intermediate property is {cash-additivity}.  A $\Lambda$-quantile with monotone $\Lambda$ functions satisfies {cash-additivity} if and only if it is a quantile; see Proposition 1 of \cite{L24}.


 Note that $q_{\Lambda}^{-}$   or  $q_{\Lambda}^{+}$ with  increasing $\Lambda$ is in general not quasi-star-shaped even when $\Lambda$ is a continuous function. 
One can see the following counterexample.
\begin{example} For $0<\alpha<1/2<\beta< 1$, consider the following continuous and increasing function $$\Lambda(x)=\alpha\id_{\{x\le1/2\}}+((\beta-\alpha)x-\beta/2+3\alpha/2)\id_{\{1/2< x< 3/2\}}+\beta\id_{\{x\ge 3/2\}}, ~~x\in\R.$$  For $t=7/4$, $\lambda=1/8$, and a random variable $X$ satisfying
$\mathbb P(X=0)=\mathbb P(X=2)=1/2$, we have $$\mathbb P(\lambda X+(1-\lambda)t=57/32)=\mathbb P(\lambda X+(1-\lambda)t=49/32)=1/2.$$ Direct computation gives $$q_{\Lambda}^{-}(X)=0,~q_{\Lambda}^{-}(\lambda X+(1-\lambda)t)=57/32,~ q_{\Lambda}^{-}(t)=7/4,$$ which implies that 
$q_{\Lambda}^{-}(\lambda X+(1-\lambda)t)>\max\{q_{\Lambda}^{-}(X),q_{\Lambda}^{-}(t)\}.$
Hence, $q_{\Lambda}^{-}$ is not quasi-star-shaped.  Note that all the above arguments still  hold if we  replace $q_{\Lambda}^{-}$ with $q_{\Lambda}^{+}$. Hence, $q_{\Lambda}^{+}$ is not quasi-star-shaped.
\end{example}

\section{Robust $\Lambda$-quantiles}\label{sec:3}
In this section, we investigate the worst-case and best-case values of $\Lambda$-quantiles over general uncertainty sets, which
is our primary objective in this paper.  For a set of distributions $\M$, we denote $F_{\M}^-(x)=\inf_{F\in\M}F(x)$, $x\in\R$ and $F_{\M}^+(x)=\sup_{F\in\M}F(x)$, $x\in\R$.  Note that both $F_{\M}^-$ and $F_{\M}^+$ are increasing functions.  Additionally, $F_{\M}^-$ is right-continuous, whereas $F_{\M}^+$ may be not. Both of them are called \emph{extremal distributions}, which may not be distribution functions.

{ Before stating our main theorem,  we first define \emph{attainability}, which is crucial for our results. 
\begin{definition}
We say that $F_{\M}^-$ (resp. $F_{\M}^+$) is \emph{attainable}  if for any $z\in\R$, there exists $F\in\M$ such that $F(z)=F_{\M}^-(z)$ (resp. $F(z)=F_{\M}^+(z)$).\end{definition}} 
The following theorem is  the main finding of this paper. 

 \begin{theorem}\label{Th:main} Let $\M$ be a set of distributions and $\Lambda:\R\to [0,1]$. Then we have
\begin{enumerate}[(i)]
 \item $\sup_{F\in\M}\tilde{q}_{\Lambda}^-(F)=\tilde{q}_{\Lambda}^-(F_{\M}^-)$ and $\inf_{F\in\M}q_{\Lambda}^+(F)=q_{\Lambda}^+(F_{\M}^+)$;
\item  $\sup_{F\in\M}\tilde{q}_{\Lambda}^+(F)=\tilde{q}_{\Lambda}^+(F_{\M}^-)$ if $F_{\M}^-$ is attainable, and $\inf_{F\in\M}q_{\Lambda}^-(F)=q_{\Lambda}^-(F_{\M}^+)$ if $F_{\M}^+$ is attainable;
 \item If $\Lambda$ is decreasing, then (i) remains true by replacing $\tilde{q}_\Lambda^-$ by $q_\Lambda^-$ and $q_\Lambda^+$  by $\tilde{q}_\Lambda^+$ and (ii) remains true by replacing $\tilde{q}_\Lambda^+$ by $q_\Lambda^+$  and  $q_\Lambda^-$ by $\tilde{q}_\Lambda^-$.
 \end{enumerate}
\end{theorem}
\begin{proof}
Case (i).  We start with the first equality. Note that $F_\M^-(x)\leq F(x),~x\in\R$ for any $F\in\M$. Hence in light of the monotonicity of $\Lambda$-quantile, we have
$\tilde{q}_{\Lambda}^-(F)\leq \tilde{q}_{\Lambda}^-(F_\M^-)$ for all $F\in\M$. This implies that $\sup_{F\in\M}\tilde{q}_{\Lambda}^-(F)\leq \tilde{q}_{\Lambda}^-(F_{\M}^-)$.

We next show the inverse inequality. Let $z=\tilde{q}_{\Lambda}^-(F_{\M}^-)$. If $z=-\infty$, $\sup_{F\in\M}\tilde{q}_{\Lambda}^-(F)=\tilde{q}_{\Lambda}^-(F_{\M}^-)=-\infty$ holds obviously.  If $z=\infty$, then there exists a sequence of $x_n$ such that $x_n\uparrow \infty$ and $F_{\M}^-(x_n)<\Lambda(x_n)$. For any fixed $x_n$, we can find $G\in\M$ such that
$F_{\M}^-(x_n)\leq G(x_n)<\Lambda(x_n)$. This implies $\sup_{F\in\M}\tilde{q}_{\Lambda}^-(F)\geq \tilde{q}_{\Lambda}^-(G)\geq x_n\to\infty$ as $n\to\infty$. Hence, it follows that
$\sup_{F\in\M}\tilde{q}_{\Lambda}^-(F)=\tilde{q}_{\Lambda}^-(F_{\M}^-)$.

 Next, we suppose that  $z\in\R$. If $F_\M^-(z)<\Lambda(z)$, there exists $G\in\M$ such that $F_\M^-(z)\leq G(z)<\Lambda(z)$. Hence, we have $\sup_{F\in\M}\tilde{q}_{\Lambda}^-(F)\geq \tilde{q}_{\Lambda}^-(G)\geq z$.
 If $F_\M^-(z)\geq \Lambda(z)$, there exists a sequence $y_n\uparrow z$ such that $F_\M^-(y_n)<\Lambda(y_n)$. For any fixed $y_n$, there exists a distribution $G\in\M$ such that $F_\M^-(y_n)\leq G(y_n)<\Lambda(y_n)$. This means $\sup_{F\in\M}\tilde{q}_{\Lambda}^-(F)\geq \tilde{q}_\Lambda^-(G)\geq y_n$. Letting $n\to\infty$, we have $\sup_{F\in\M}\tilde{q}_{\Lambda}^-(F)\geq z$.  We { have established} the inverse inequality.

 We now focus on the second equality. Observe that $F_\M^+(x)\geq F(x),~x\in\R$ for any $F\in\M$. It follows from the monotonicity of $\Lambda$-quantile that $q_{\Lambda}^+(F_{\M}^+)\leq q_{\Lambda}^+(F)$ for all $F\in\M$. Consequently, $q_{\Lambda}^+(F_{\M}^+)\leq \inf_{F\in\M}q_{\Lambda}^+(F)$.

 Let $z=q_{\Lambda}(F_{\M}^+)$. If $z=\infty$, then $q_{\Lambda}^+(F_{\M}^+)=\inf_{F\in\M}q_{\Lambda}^+(F)$ holds trivially. If $z=-\infty$, there exists a sequence of $x_n\downarrow -\infty$ such that  $F_{\M}^+(x_n)> \Lambda(x_n)$.  For any fixed $x_n$, there exists $G\in\M$ such that $F_{\M}^+(x_n)\geq G(x_n)>\Lambda(x_n)$. This implies that $\inf_{F\in\M}q_{\Lambda}^+(F)\leq q_{\Lambda}^+(G)\leq x_n\to -\infty$ as $n\to\infty$. Hence, we obtain $\inf_{F\in\M}q_{\Lambda}^+(F)=q_{\Lambda}(F_{\M}^+)$.

  We suppose that  $z\in\R$. If $F_{\M}^+(z)>\Lambda(z)$, there exists $G\in\M$ such that $F_{\M}^+(z)\geq G(z)>\Lambda(z)$, implying $\inf_{F\in\M} q_{\Lambda}^+(F)\leq q_{\Lambda}^+(G)\leq z$. If $F_{\M}^+(z)\leq \Lambda(z)$, there exists a sequence of $y_n\downarrow z$ such that $F_{\M}^+(y_n)>\Lambda(y_n)$. For fixed $y_n$, there exists $G\in\M$ such that $F_{\M}^+(y_n)\geq G(y_n)>\Lambda(y_n)$, which implies that 
  $\inf_{F\in\M}q_{\Lambda}^+(F)\leq q_{\Lambda}^+(G)\leq y_n\to z$ as $n\to\infty$. We { have established} the inverse inequality.

  Case (ii).  For the first equality, note that $\sup_{F\in\M}\tilde{q}_{\Lambda}^+(F)\leq \tilde{q}_{\Lambda}^+(F_{\M}^-)$ follows from the same reasoning as in the proof of case (i). We only need to show the inverse inequality.

  Denote $\tilde{q}_{\Lambda}^+(F_{\M}^-)$ by $z$. If $z=-\infty$, then $\sup_{F\in\M}\tilde{q}_{\Lambda}^+(F)=\tilde{q}_{\Lambda}^-(F_{\M}^-)=-\infty$ holds obviously.  If $z=\infty$, then there exists a sequence of $x_n$ such that $x_n\uparrow \infty$ and $F_{\M}^-(x_n)\leq \Lambda(x_n)$. For any fixed $x_n$, using the attainability of $F_{\M}^-$, there exists $G\in\M$ such that
$G(x_n)=F_{\M}^-(x_n)\leq \Lambda(x_n)$. This implies $\sup_{F\in\M}\tilde{q}_{\Lambda}^+(F)\geq \tilde{q}_{\Lambda}^+(G)\geq x_n\to\infty$ as $n\to\infty$. Hence, we have
$\sup_{F\in\M}\tilde{q}_{\Lambda}^+(F)=\tilde{q}_{\Lambda}^+(F_{\M}^-)$.

 Next, we suppose that  $z\in\R$. If $F_\M^-(z)\leq \Lambda(z)$, there exists $G\in\M$ such that $G(z)=F_\M^-(z)\leq \Lambda(z)$. Hence, we have  $\sup_{F\in\M}\tilde{q}_{\Lambda}^+(F)\geq \tilde{q}_{\Lambda}^+(G)\geq z$.
 If $F_\M^-(z)>\Lambda(z)$, there exists a sequence $y_n\uparrow z$ such that $F_\M^-(y_n)\leq \Lambda(y_n)$. For any fixed $y_n$, there exists a distribution $G\in\M$ such that $G(y_n)=F_\M^-(y_n)\leq \Lambda(y_n)$. This means that $\sup_{F\in\M}\tilde{q}_{\Lambda}^+(F)\geq \tilde{q}_\Lambda^+(G)\geq y_n$. Letting $n\to\infty$, we have $\sup_{F\in\M}\tilde{q}_{\Lambda}^+(F)\geq z$.  We { have established} the inverse inequality.

  The proof of the second equality of (ii) follows a similar argument as that of  the second equality of Case (i) by applying the attainability of $F_\M^+$. We omit the details.

Case (iii).  
If $\Lambda$ is decreasing, by (i) of Proposition \ref{prop:1}, we have $q_{\Lambda}^{-}(F)=\tilde{q}_{\Lambda}^{-}(F)$ and $q_{\Lambda}^{+}(F)= \tilde{q}_{\Lambda}^{+}(F)$.  The conclusion in (iii) is directly implied by the conclusions in (i)-(ii).
\end{proof}

 { Theorem \ref{Th:main} shows that eight robust $\Lambda$-quantiles can be calculated using the extremal distributions. It is worth mentioning that they may require different conditions to guarantee that the calculation is correct: (i) of Theorem \ref{Th:main} holds for all $\Lambda$ functions and uncertainty sets $\mathcal M$; (ii) of Theorem \ref{Th:main} holds if the extremal distributions satisfy attainability; (iii) of Theorem \ref{Th:main} holds only for decreasing $\Lambda$ and two of them also require the attainability of the extremal distributions. Roughly speaking, only four of robust $\Lambda$-quantiles can be calculated using the extremal distributions with general $\Lambda$ functions; the rest holds true only for  decreasing $\Lambda$ functions. If the required conditions are not satisfied, then the $\Lambda$-quantiles of the extremal distributions can only be the bounds of the robust $\Lambda$-quantiles. Some counterexamples  will be discussed later.}

 Finding the extremal distributions is a problem in probability theory with a long history; see \cite{M81} and \cite{R82} for early results, and see \cite{L18}, \cite{BLLW20} and \cite{PWW20} for the recent results. Roughly speaking, the main message of our results is that if we know the {extremal distributions} and the assumptions are satisfied, applying Theorem \ref{Th:main}, we could obtain the robust $\Lambda$-quantiles immediately. This enables us to obtain the robust $\Lambda$-quantiles with the aid of many results on robust-quantiles in the literature.


  Applying Theorem \ref{Th:main}, we obtain the following conclusion for robust quantiles, which indicates the relation between robust-quantiles and {extremal distributions}.

\begin{corollary}\label{coro:1} For $\alpha\in (0,1)$, we have
\begin{enumerate}[(i)]
 \item $\sup_{F\in\M}q_{\alpha}^-(F)=q_{\alpha}^-(F_{\M}^-)$ and $\inf_{F\in\M}q_{\alpha}^+(F)=q_{\alpha}^+(F_{\M}^+)$.
\item  $\sup_{F\in\M}q_{\alpha}^+(F)=q_{\alpha}^+(F_{\M}^-)$ if $F_{\M}^-$ is attainable, and $\inf_{F\in\M}q_{\alpha}^-(F)=q_{\alpha}^-(F_{\M}^+)$ if $F_{\M}^+$ is attainable.
  \end{enumerate}
\end{corollary}

  The conclusion in (iii) of Theorem  \ref{Th:main} does not hold for other $\Lambda$ functions even for increasing $\Lambda$ functions.  We can see this  from the following example.
\begin{example} Let $$\Lambda(x)=\left\{\begin{array}{cc}
1/4,& x<0,\\
(1+x)/4,& 0\leq x<1,\\
1/2,& x\geq 1,
\end{array}\right.$$ and $\M=\{F_1, F_2\}$ with  $$F_1(x)=\left\{\begin{array}{cc}
0,& x<0,\\
1/3, & 0\leq x<1,\\
1,& x\geq 1,
\end{array}\right.~\text{and}~ F_2(x)=\left\{\begin{array}{cc}
0,& x<1/2,\\
1,& x\geq 1/2.
\end{array}\right.$$
Direct computation gives $$F_{\M}^-(x)=\left\{\begin{array}{cc}
0,& x<1/2,\\
1/3, & 1/2\leq x<1,\\
1,& x\geq 1,
\end{array}\right.$$
$q_{\Lambda}^-(F_1)=0$, $q_{\Lambda}^-(F_2)=1/2$ and $q_{\Lambda}^-(F_{\M}^-)=1$. Hence
$\sup_{F\in\M}q_{\Lambda}^-(F)=1/2<1=q_{\Lambda}^-(F_{\M}^-)$. We can similarly construct counterexamples for other three $\Lambda$-quantiles in (iii) of Theorem \ref{Th:main} for increasing $\Lambda$ functions.

\end{example}
{ The conclusion in  (ii) of Theorem \ref{Th:main} requires the attainability of $F_{\M}^-$ or $F_{\M}^+$}.  If   we have  finite candidates in $\M$, i.e., $\M=\{F_1,\dots, F_n\}$. Then $F_\M^-(x)=\bigwedge_{i=1}^n F_i(x)$ and $F_\M^+(x)=\bigvee_{i=1}^n F_i(x)$, and both $F_\M^-$ and $F_\M^+$ are attainable, where $\bigwedge_{i=1}^n x_i=\min(x_1,\dots, x_n)$ and  $\bigvee_{i=1}^n x_i=\max(x_1,\dots, x_n)$. Applying Theorem \ref{Th:main}, we arrive at the following results.
\begin{corollary}
  Let $\M=\{F_1,\dots, F_n\}$ and $\Lambda:\R\to [0,1]$. The following statements hold.
\begin{enumerate}[(i)]
 \item $\sup_{F\in\M}\tilde{q}_{\Lambda}^-(F)=\tilde{q}_{\Lambda}^-(\bigwedge_{i=1}^n F_i)$, $\inf_{F\in\M}q_{\Lambda}^+(F)=q_{\Lambda}^+(\bigvee_{i=1}^n F_i)$,  $\sup_{F\in\M}\tilde{q}_{\Lambda}^+(F)=\tilde{q}_{\Lambda}^+(\bigwedge_{i=1}^n F_i)$, and $\inf_{F\in\M}q_{\Lambda}^-(F)=q_{\Lambda}^-(\bigvee_{i=1}^n F_i)$.
 \item If $\Lambda$ is decreasing, then we have $\sup_{F\in\M}q_{\Lambda}^-(F)=q_{\Lambda}^-(\bigwedge_{i=1}^n F_i)$, $\inf_{F\in\M}\tilde{q}_{\Lambda}^+(F)=\tilde{q}_{\Lambda}^+(\bigvee_{i=1}^n F_i)$,
      $\sup_{F\in\M}q_{\Lambda}^+(F)=q_{\Lambda}^+(\bigwedge_{i=1}^n F_i)$,  and $\inf_{F\in\M}\tilde{q}_{\Lambda}^-(F)=\tilde{q}_{\Lambda}^-(\bigvee_{i=1}^n F_i)$.
 \end{enumerate}
\end{corollary}
 The property of attainability may not hold if  the candidates in $\mathcal M$ are infinite. We have the following counterexample.
\begin{example} Let $$\Lambda(x)=\left\{\begin{array}{cc}
3/4,& x<0,\\
1/2,& 0\leq x<1,\\
1/4,& x\geq 1,
\end{array}\right.$$ and $\M=\{F_n, n\geq 1\}$ with  $$F_n(x)=\left\{\begin{array}{cc}
0,& x<0,\\
1/2+1/(2n), & 0\leq x<1,\\
1,& x\geq 1.
\end{array}\right.$$
Direct computation gives $$F_{\M}^-(x)=\left\{\begin{array}{cc}
0,& x<0,\\
1/2, & 0\leq x<1,\\
1,& x\geq 1,
\end{array}\right.$$
$\tilde{q}_{\Lambda}^+(F_n)=0,~n\geq 1$ and
$\tilde{q}_{\Lambda}^+(F_{\M}^-)=1$.
 Hence, $\sup_{F\in\M}\tilde{q}_{\Lambda}^+(F)<\tilde{q}_{\Lambda}^+(F_{\M}^-)$. We can similarly construct examples to show that $\inf_{F\in\M}q_{\Lambda}^-(F) >q_{\Lambda}^-(F_{\M}^+)$ if $F_{\M}^+$ is not attainable.
\end{example}
The above example shows that if the attainability fails to hold, then (ii) of Theorem \ref{Th:main} may not be valid. However, we still have the following conclusions: $$\sup_{F\in\M}\tilde{q}_{\Lambda}^+(F)\leq \tilde{q}_{\Lambda}^+(F_{\M}^-)~\text{and}~ \inf_{F\in\M}q_{\Lambda}^-(F)\geq q_{\Lambda}^-(F_{\M}^+).$$  Hence, both $\tilde{q}_{\Lambda}^+(F_{\M}^-)$ and $q_{\Lambda}^-(F_{\M}^+)$ are still useful because they offer the upper and lower bounds of the corresponding robust $\Lambda$-quantiles.

Compared  with   the  attainability of $F_{\M}^+$, the attainability of $F_{\M}^-$  may not be easy to check as it is difficult to employ the limit argument; see the examples of the three uncertainty sets in Section \ref{sec:4}.   Next, we provide an alternative assumption that is much easier to check if $\Lambda$ is decreasing. Let $\widehat{F}_{\M}^-(x)=\inf_{F\in\M} F(x-)$ for all $x\in\R$. We say $\widehat{F}_{\M}^-$ is attainable if for any $x\in\R$, there exists $F\in\M$ such that $F(x-)=\widehat{F}_{\M}^-(x)$. { The attainability of $\widehat{F}_{\M}^-$ for three specific uncertainty sets will be discussed in Propositions \ref{lemma:1}, \ref{lemma:2} and \ref{prop:aggregationF} in Section \ref{sec:4}.}

With the attainability of $\widehat{F}_{\M}^-$ and the monotonicity of $\Lambda$, we obtain the first equality in (ii) of Theorem \ref{Th:main}.
\begin{proposition}\label{prop:attainable} Let $\M$ be a set of distributions and $\Lambda:\R\to [0,1]$. If $\Lambda$ is decreasing and $\widehat{F}_{\M}^-$ is continuous and attainable, we have
$\sup_{F\in\M}\tilde{q}_{\Lambda}^+(F)=\tilde{q}_{\Lambda}^+(F_{\M}^-)$.
\end{proposition}
\begin{proof}
Let $\widehat{F}(x)=F(x-)$ for all $x\in\R$. Replacing $F$ by $\widehat{F}$ and $F_{\M}^-$ by $\widehat{F}_{\M}^-$ in the proof of (ii) of Theorem \ref{Th:main} and using the attainability of $\widehat{F}_{\M}^-$, we obtain the conclusion that $\sup_{F\in\M}\tilde{q}_{\Lambda}^+(\widehat{F})=\tilde{q}_{\Lambda}^+(\widehat{F}_{\M}^-)$. Note that for any $\epsilon>0$, $\widehat{F}_{\M}^-(x)\leq F_{\M}^-(x)\leq \widehat{F}_{\M}^-(x+\epsilon)$. Hence, by the continuity of $\widehat{F}_{\M}^-$, we have {$\widehat{F}_{\M}^-=F_{\M}^-$}. Moreover, as $\Lambda$ is decreasing, we have $\tilde{q}_{\Lambda}^+(F)=\tilde{q}_{\Lambda}^+(\widehat{F})$. Combing all the above conclusions, we have $\sup_{F\in\M}\tilde{q}_{\Lambda}^+(F)=\tilde{q}_{\Lambda}^+(F_{\M}^-)$.
\end{proof}

Next, we discuss the existence of the worst or best distribution for the robust $\Lambda$-quantiles.

\begin{proposition} Let $\M$ be a set of distributions and $\Lambda:\R\to [0,1]$. Then we have the following conclusion.
\begin{enumerate}[(i)]
\item If $\Lambda$ is decreasing and left-continuous, and $\widehat{F}_{\M}^-$ is continuous and attainable, and $\tilde{q}_{\Lambda}^+(F_{\M}^-)\in\R$, then there exists $F^*\in\M$ such that $\tilde{q}_{\Lambda}^+(F^*)=\sup_{F\in\M}\tilde{q}_{\Lambda}^+(F)$.
\item If $\Lambda$ is right-continuous, $F_{\M}^+$ is attainable and right-continuous and  $q_{\Lambda}^-(F_{\M}^+)\in\R$,  then there exists $F^*\in\M$  such that $q_{\Lambda}^-(F^*)=\inf_{F\in\M}q_{\Lambda}^-(F)$.
\end{enumerate}
\end{proposition}
\begin{proof} (i)
In light of Proposition \ref{prop:attainable}, we have $\sup_{F\in\M}\tilde{q}_{\Lambda}^+(F)=\tilde{q}_{\Lambda}^+(F_{\M}^-)$. For the sake of  simplicity, we denote $\tilde{q}_{\Lambda}^+(F_{\M}^-)$ by $z$. By the definition, we have either $F_{\M}^-(z)\leq \Lambda(z)$ or there exists a sequence $x_n\uparrow z$ such that $F_{\M}^-(x_n)\leq \Lambda(x_n)$. Note that $F_{\M}^-=\widehat{F}_{\M}^-$. Hence, $F_{\M}^-$ is continuous. For the second case, using the continuity of $F_{\M}^-$ and the left-continuity of $\Lambda$, letting $n\to\infty$, we have $F_{\M}^-(z)\leq \Lambda(z)$.  As $\widehat{F}_{\M}^-$ is attainable, there exists {$F^*\in\M$ such that $\widehat{F}^*(z)=\widehat{F}_{\M}^-(z)=F_{\M}^-(z)$} with $\widehat{F}^*(z)=F^*(z-)$. This implies that $\widehat{F}^*(z)\leq \Lambda(z)$; hence, $\tilde{q}_{\Lambda}^+(\widehat{F}^*)\geq z=\sup_{F\in\M}\tilde{q}_{\Lambda}^+(F)$. Note that the decrease of $\Lambda$ implies that  $\tilde{q}_{\Lambda}^+(F^*)=\tilde{q}_{\Lambda}^+(\widehat{F}^*)$. Hence, we have $\tilde{q}_{\Lambda}^+(F^*)\geq \sup_{F\in\M}\tilde{q}_{\Lambda}^+(F)$, which implies that $\tilde{q}_{\Lambda}^+(F^*)=\sup_{F\in\M}\tilde{q}_{\Lambda}^+(F)$.

(ii) First, it follows from Theorem \ref{Th:main} that $\inf_{F\in\M}q_{\Lambda}^-(F)=q_{\Lambda}^-(F_{\M}^+)$. Let $z=q_{\Lambda}^-(F_{\M}^+)$.  By definition, we have either $F_{\M}^+(z)\geq \Lambda(z)$ or there exists a sequence of $x_n\downarrow z$ such that $F_{\M}^+(x_n)\geq \Lambda(x_n)$. For the second case,  it follows from the right-continuity of $F_{\M}^+$ and $\Lambda$ that $F_{\M}^+(z)\geq \Lambda(z)$. The attainability of $F_{\M}^+$ implies  $F^*\in\M$ exists  such that $F^*(z)=F_{\M}^+(z)\geq \Lambda(z)$. This implies $q_{\Lambda}^-(F^*)\leq z=\inf_{F\in\M}q_{\Lambda}^-(F)$. Hence, we have $q_{\Lambda}^-(F^*)=\inf_{F\in\M}q_{\Lambda}^-(F)$.
\end{proof}
Note that for other $\Lambda$-quantiles,  the best-case or worst-case  distribution may not exist, even for the case when $\Lambda=\alpha$. For instance, for the uncertainty set characterized by moment constraints, the optimal distributions for $\sup_{F\in\M}q_{\alpha}^-(F)$ and $\inf_{F\in\M}q_{\alpha}^+(F)$ do not exist; see e.g., Corollary 4.1 of \cite{BPV24}.

Finally, we offer alternative expressions of robust $\Lambda$-quantiles using \eqref{eq:alt_pre} and \eqref{eq:alt_pre2} when $\Lambda$ is decreasing. The expressions can sometimes facilitate the calculation of robust $\Lambda$-quantiles more conveniently;  see Sections \ref{sec:4.1} and \ref{sec:4.2} for details.
\begin{proposition}\label{prop:2}
  If $\Lambda$ is decreasing, then we have the following representations.
  \begin{itemize} \item[(i)] $\sup_{F\in\M}{q}_{\Lambda}^-(F)=  \sup_{x\in \mathbb{R}}    \{q_{\Lambda(x)}^{-} ( F^-_{\mathcal M}) \wedge x\},$ and $\inf_{F\in\M}q_{\Lambda}^+(F)=\inf_{x\in \mathbb{R}}    \{q_{\Lambda(x)}^{+} ( F^+_{\mathcal M}) \vee x\}$.
\item[(ii)] If  $F_{\M}^-$ is attainable, then   $\sup_{F\in\M}{q}_{\Lambda}^+(F)=  \sup_{x\in \mathbb{R}}    \{q_{\Lambda(x)}^{+} ( F^-_{\mathcal M}) \wedge x\}$ , and if $F_{\M}^+$ is attainable, then  $\inf_{F\in\M}q_{\Lambda}^-(F)=\inf_{x\in \mathbb{R}}\{q_{\Lambda(x)}^{-} ( F^+_{\mathcal M}) \vee x\}$.
\end{itemize}
\end{proposition}

\begin{proof} The results directly follow from Theorem \ref{Th:main}, and Equations \eqref{eq:alt_pre}-\eqref{eq:alt_pre2}. However, we will present an alternative proof to demonstrate the connection between robust $\Lambda$-quantiles and  quantiles. We take the first one as an example. If  $\Lambda$ is decreasing, by \eqref{eq:alt_pre}, we have  \begin{align*}
 \sup_{F\in\M}  q_{\Lambda}^{-} (F) &=\sup_{F\in\M}   \left\{\sup_{x\in \mathbb{R}} \left\{q_{\Lambda(x)}^{-} (F) \wedge x\right\}\right\}\\& = \sup_{x\in \mathbb{R}}    \left\{\sup_{F\in\M}  \left\{q_{\Lambda(x)}^{-} (F) \wedge x\right\}\right\} = \sup_{x\in \mathbb{R}}    \left\{\left\{\sup_{F\in\M}  q_{\Lambda(x)}^{-} (F)\right\} \wedge x\right\}. \end{align*}
It follows from Corollary \ref{coro:1} that $\sup_{F\in\M}  q_{\Lambda(x)}^{-} (F)= q_{\Lambda(x)}^{-} (F^-_{\mathcal{M}})$ for any $x\in\R$.
 Hence, we have $\sup_{F\in\M}  q_{\Lambda}^{-} (F) = \sup_{x\in \mathbb{R}}    \left\{q_{\Lambda(x)}^{-} ( F^-_{\mathcal M}) \wedge x\right\}.$
\end{proof}

\section{Specific uncertainty sets}\label{sec:4}
In this section, our attention is directed towards the uncertainty sets characterized by three distinct and widely used constraints: (a) moment constraints, (b) probability constraints via the Wasserstein distance; and (c) marginal constraints in risk aggregation. Furthermore, our results are applied to the portfolio selection problem. By using Theorem \ref{Th:main}, we can clearly see that the robust $\Lambda$-quantiles can be derived based on the expressions of robust quantiles. 
Denote by   $\mathcal{M}_p$  the set of  all cdfs on $\mathbb{R}$ with a finite $p$-th moment and  $\mathcal{M}_p\left(\mathbb{R}^n\right)$  the set of all cdfs  on $\mathbb{R}^n$ with a finite $p$-th moment in each component.



\subsection{Moment constraints}\label{sec:4.1}
We  first consider the case with only information on the moments of $F\in\mathcal M$.
For $p>1, m \in \mathbb{R}$, and $v>0$,  let
\begin{equation}\label{def:M_p}
\mathcal{M}_p( m, v)=\left\{F \in \mathcal{M}_p: \int_{0}^{1}F^{-1}(t)\d t=m, \int_{0}^{1}|F^{-1}(t)-m|^p\d t \leqslant v^p\right\}.\end{equation}
Note that $\mathcal{M}_p( m, v)$  imposes the constraints on the mean and the $p$-th central moment.  
 The worst-case and best-case  values of some risk measures over  $\M_p(m,v)$, especially for $p=2$, have been extensively studied in the literature; see e.g., \cite{GOO03}, \cite{ZF09},  \cite{CHZ11}, \cite{L18}, \cite{PWW20}, \cite{SZ23}, \cite{SZ24}, \cite{CLM23}  and the references therein.

{For an increasing function $g:(0,1)\to\R$, define $g_+^{-1}(x)=\inf\{t\in(0,1): g(t)>x\}$ with the convention that $\inf\emptyset=1$.}
In the following results, we derive the {extremal distributions} over the moment uncertainty sets, i.e., $F_{\M_p(m,v)}^-$ and $F_{\M_p(m,v)}^+$.
\begin{proposition}\label{lemma:1}
 For $p>1, m \in \mathbb{R}$ and $v>0$, we have
 \begin{align}\label{eq:mv_1}
 F_{\M_p(m,v)}^-(x)=l_{+}^{-1}(x),~~F_{\M_p(m,v)}^+(x)=u_+^{-1}(x),~x\in\R,
 \end{align}
 where, for $\alpha\in (0,1)$,
 \begin{align*}
l(\alpha)= m+v \alpha\left(\alpha^p(1-\alpha)+(1-\alpha)^p \alpha\right)^{-1 / p},~
u(\alpha)= m-v(1-\alpha)\left(\alpha^p(1-\alpha)+(1-\alpha)^p \alpha\right)^{-1 / p}.
\end{align*} 
In particular, for $p=2$, we have    $$ { F^-_{\M_2(m,v)}(x)=\frac{(x-m)^2}{v^2+(x-m)^2} \id_{\{x\geq m\}}}, ~~ \text{and}  ~~F^+_{\M_2(m,v)}(x)=  \frac{v^2}{v^2+(m-x)^2}\id_{\{x \leq m\}}+\id_{\{x>m\}}.$$
Moreover, both $\widehat{F}_{\M_p(m,v)}^-$ and $F_{\M_p(m,v)}^+$ are continuous and attainable.
\end{proposition}
\begin{proof} By Corollary 1  of \cite{PWW20}, for $\alpha \in(0,1), p>1, m \in \mathbb{R}$ and $v>0$, we have
$$
\sup _{F \in \mathcal{M}_p( m, v)} q_\alpha^{-}(F)=m+v \alpha\left(\alpha^p(1-\alpha)+(1-\alpha)^p \alpha\right)^{-1 / p}:=l(\alpha).
$$
Moreover, in light of (i) of Corollary \ref{coro:1}, we have $\sup _{F \in \mathcal{M}_p( m, v)} q_\alpha^{-}(F)=q_{\alpha}^-(F_{\M_p(m,v)}^-)$. Hence, $q_{\alpha}^-(F_{\M_p(m,v)}^-)=l(\alpha)$ for $\alpha\in (0,1)$.  Note that $F_{\M_p(m,v)}^-(\cdot)$ is right-continuous.  Therefore, $F_{\M_p(m,v)}^-(x)=l_{+}^{-1}(x)$ for $x\in\R$. One can easily check that $l$ is strictly increasing and continuous over $(0,1)$. Thus, we have that $F_{\M_p(m,v)}^-$ is continuous over $\R$, which further implies  that $\widehat{F}_{\M_p(m,v)}^-$ is continuous over $\R$.

  Again, by  Corollary 1  of \cite{PWW20}, we have $$
\inf_{F \in \mathcal{M}_p( m, v)} q^+_\alpha(F)=m-v(1-\alpha)\left(\alpha^p(1-\alpha)+(1-\alpha)^p \alpha\right)^{-1 / p}:=u(\alpha).$$
Moreover, it follows from (i) of Corollary \ref{coro:1}, we have $\inf _{F \in \mathcal{M}_p( m, v)} q_\alpha^{+}(F)=q_{\alpha}^+(F_{\M_p(m,v)}^+)$. By combining the above two equations, we have $q_{\alpha}^+(F_{\M_p(m,v)}^+)=u(\alpha)$ for $\alpha\in (0,1)$. Note that $u(\alpha)$ is strictly increasing and continuous over $(0,1)$. Hence, $F_{\M_p(m,v)}^+(x)=u_+^{-1}(x)$ for $x\in\R$ and $F_{\M_p(m,v)}^+$ is continuous.

For $\alpha\in(0,1)$ and $p=2$,  we have
$l(\alpha)=m+v \sqrt{{\alpha}/{(1-\alpha)}}$ and $u(\alpha)=m-v \sqrt{(1-\alpha)/\alpha}$.
We obtain the desired results by computing $l_+^{-1}$ and $u_+^{-1}$.

Next, we show the attainability of $\widehat{F}_{\M_p(m,v)}^-$ and $F_{\M_p(m,v)}^+$. For any $x\in\R$, there exists $F_k\in \M_p(m,v)$ such that
$\lim_{k\to\infty} F_k(x-)=\widehat{F}_{\M_p(m,v)}^-(x)$. By Chebyshev's inequality, we have $\sup_{k\geq 1}(1-F_k(z))\leq \frac{v^p}{|z-\mu|^p}$ for $z>\mu$ and $\sup_{k\geq 1}F_k(z)\leq \frac{v^p}{|z-\mu|^p}$ for $z<\mu$. Hence, $\{F_k,~k\geq 1\}$ is tight. It follows from Theorem 25.10 of \cite{B95}) that there exist a subsequence $\{F_{k_l},~l\geq 1\}$ and a distribution $F$ such that $F_{k_l}\to F$ weakly as $l\to\infty$.  This implies that $F_{k_l}^{-1}(t) \to F^{-1}(t)$ over $(0,1)$ a.e..

Next, we show $F\in \M_p(m,v)$. By Fatou's lemma, we have
$$ \int_{0}^{1}|F^{-1}(t)-m|^p\d t=\int_{0}^{1}\liminf_{l\to\infty}|F_{k_l}^{-1}(t)-m|^p\d t\leq \liminf_{l\to\infty}\int_{0}^{1}|F_{k_l}^{-1}(t)-m|^p\d t\leqslant v^p.$$
Given that  $\sup_{l\geq 1}\int_{0}^{1}|F_{k_l}^{-1}(t)|^p\d t\leq (v+|m|)^p$, we have for $y>0$
$$\int_{\{t: |F_{k_l}^{-1}(t)|\geq y\}}|F_{k_l}^{-1}(t)|\d t\leq \int_{\{t: |F_{k_l}^{-1}(t)|\geq y \}}|F_{k_l}^{-1}(t)|^p y^{1-p}\d t\leq (v+|m|)^py^{1-p},$$
which implies $\lim_{y\to\infty}\sup_{l\geq 1}\int_{\{t: |F_{k_l}^{-1}(t)|\geq y\}}|F_{k_l}^{-1}(t)|\d t=0.$ It means that $\{F_{k_l}^{-1}, l\geq 1\}$ is uniformly integrable. Hence, we have $\int_{0}^{1}F^{-1}(t)\d t=\lim_{l\to\infty}\int_{0}^{1}F_{k_l}^{-1}(t)\d t=m$.
Consequently, $F\in \M_p(m,v)$.

Note that there exists a sequence $\epsilon_n\downarrow 0$  as $n\to \infty$ such that $F$ is continuous at all $x-\epsilon_n$. Using the weak convergence of $F_{k_l}$ to $F$, we have $F(x-\epsilon_n)=\lim_{l\to\infty}F_{k_l}(x-\epsilon_n)\leq \lim_{l\to\infty}F_{k_l}(x-)=\widehat{F}_{\M_p(m,v)}^-(x)$. Letting $n\to\infty$, we have $F(x-)\leq\widehat{F}_{\M_p(m,v)}^-(x)$, which implies
$F(x-)=\widehat{F}_{\M_p(m,v)}^-(x)$. Therefore, we conclude that $\widehat{F}_{\M_p(m,v)}^-$ is attainable.

For any $x\in\R$, there exists a sequence  $F_k\in \M_p(m,v)$ such that
$\lim_{k\to\infty} F_k(x)=F_{\M_p(m,v)}^+(x)$. Similarly as in the proof of the attainability of $\widehat{F}_{\M_p(m,v)}^-$, there exists a subsequence $\{F_{k_l},~l\geq 1\}$ and $F\in \M_p(m,v)$ such that $F_{k_l}\to F$ weakly as $l\to\infty$. Moreover, we could find $\epsilon_n\downarrow 0$ such that $F$ is continuous at all $x+\epsilon_n$. Using the fact that $F_{k_l}\to F$ weakly as $l\to\infty$, we have
$F(x+\epsilon_n)=\lim_{l\to\infty} F_{k_l}(x+\epsilon_n)\geq \lim_{l\to\infty} F_{k_l}(x)=F_{\M_p(m,v)}^+(x)$. Letting $n\to\infty$, we obtain $F(x)\geq F_{\M_p(m,v)}^+(x)$, which implies $F(x)=F_{\M_p(m,v)}^+(x)$. Hence, $F_{\M_p(m,v)}^+$ is attainable.
\end{proof}


Combing  Theorem \ref{Th:main}, Propositions \ref{prop:attainable} and \ref{lemma:1}, we immediately arrive at the following result.

 \begin{theorem}\label{thm:2}
Let $F_{\M_p(m,v)}^-$ and $F_{\M_p(m,v)}^+$ be given by \eqref{eq:mv_1} and $\Lambda:\R\to [0,1]$.  Then we have
\begin{enumerate}[(i)]
 \item $\sup_{F\in \M_p(m,v)}\tilde{q}_{\Lambda}^-(F)=\tilde{q}_{\Lambda}^-(F_{\M_p(m,v)}^-)$, $\inf_{F\in \M_p(m,v)}q_{\Lambda}^+(F)=q_{\Lambda}^+(F_{\M_p(m,v)}^+)$, and     $\inf_{F\in\M_p(m,v)}q_{\Lambda}^-(F)=q_{\Lambda}^-(F_{\M_p(m,v)}^+)$;
\item  If $\Lambda$ is decreasing, then (i) remains true by replacing $\tilde{q}_\Lambda^-$ by $q_\Lambda^-$, $q_\Lambda^+$  by $\tilde{q}_\Lambda^+$ and  $q_\Lambda^-$ by $\tilde{q}_\Lambda^-$; Moreover, $\sup_{F\in\M_p(m,v)}\tilde{q}_{\Lambda}^+(F)=\sup_{F\in\M_p(m,v)}q_{\Lambda}^+(F)=\tilde{q}_{\Lambda}^+(F_{\M_p(m,v)}^-)$.
 \end{enumerate}
 \end{theorem}

For the case where $p=2$,  we obtain the explicit formulas for  $F^-_{\M_2(m,v)}$ and $F^+_{\M_2(m,v)}$.  For general case  where $p\neq2$, the values of   $F^-_{\M_{(p, m, v)}}$ and $F^+_{\M_{(p, m, v)}}$  can be computed numerically using \eqref{eq:mv_1}.  For the particular case in which  $\Lambda$ is decreasing, it is more convenient to compute  the robust $\Lambda$-quantiles using the expressions in Proposition \ref{prop:2} since we do not need to compute the inverse functions $l$ and $u$ to obtain $F_{\M_p{( m, v)}}^-$ and $F_{\M_p{(m, v)}}^+$.

 In the following example, we compare the $\Lambda$-quantiles of normal,  exponential  and uniform distributions  with  their worst-case and best-case  values over $\M_2(m,v)$,  respectively.  
\begin{example}\label{exm:4}  Let $p=2$, $m=v=1$, and  $\Lambda_{(\alpha,\beta;z)}: x\mapsto \alpha \id_{\{x< z\}} + \beta \id_{\{x\geq z\}}$.   Using  Theorem \ref{thm:2} and Proposition \ref{prop:2} , the numerical results of robust $\Lambda$-quantiles are displayed in Tables \ref{tab:1} and \ref{tab:2} for two different $\Lambda$ functions.
\begin{table}[h]
\def\arraystretch{1} \begin{center}  \caption{Robust  $\Lambda$-quantiles   for an increasing $\Lambda$ with $\mathcal M=\mathcal M_2(m,v)$ }   \label{tab:1}  \begin{tabular}{c|cccccc}   $\Lambda_{(0.8,0.95;2)}$ & $ q_{\Lambda}^{-}  $ &   $q_{\Lambda}^{+} $&$ \tilde{q}_{\Lambda}^{-}  $&$ \tilde{q}_{\Lambda}^{+}$   \\ \hline
 $ \mathrm{N}(1,1)$ &   1.84  &  1.84  &  2.64 & 2.64   \\
   $ \mathrm{exp}(1)$    &1.61 & 1.61& 2.98 & 2.98  \\
   $ U[1-\sqrt{3},1+\sqrt{3} ]$  &2.57  &2.57 & 2.57& 2.57   \\\hline
  Best-case &   0.51  &  0.51& \textbf { 0.51} &\textbf {0.51}   \\
 Worst-case &  \textbf {5.36}  &\textbf {5.36}  &5.36&\textbf {5.36}  \\ \hline  \end{tabular}  \end{center}\end{table}
 \begin{table}[h]
\def\arraystretch{1} \begin{center}  \caption{ Robust  $\Lambda$-quantiles   for a decreasing  $\Lambda$ with $\mathcal M=\mathcal M_2(m,v)$}  \label{tab:2}  \begin{tabular}{c|cccccc}   $\Lambda_{(0.95,0.8;2)}$ & $ q_{\Lambda}^{-}  $ &   $q_{\Lambda}^{+} $&$ \tilde{q}_{\Lambda}^{-}  $&$ \tilde{q}_{\Lambda}^{+}$   \\ \hline
 $\mathrm{N}(1,1)$ &   2.00  &  2.00  &   2.00 &  2.00  \\
   $  \mathrm{exp}(1)$    &2.00& 2.00& 2.00& 2.00 \\
    $ U[1-\sqrt{3},1+\sqrt{3} ]$  &2.05  &2.05 & 2.05 & 2.05  \\\hline
  Best-case & 0.90 & 0.90   & 0.90   &  0.90    \\
 Worst-case &  3.00  &3.00  &3.00  &3.00  \\ \hline  \end{tabular}  \end{center}\end{table}

For an increasing $\Lambda$, it can be observed from Table \ref{tab:1}  that $q_{\Lambda}^{-}$ is strictly smaller than $\tilde{q}_{\Lambda}^{-}$, and $q_{\Lambda}^{+}$ is strictly smaller than $\tilde{q}_{\Lambda}^{+}$. This finding  verifies the conclusion in (i) of Proposition \ref{prop:1}.  Note that the bold numbers in Table 1 may not be  the true  best-case  or worst-case values; instead, they represent the lower bounds of the worst-case values or  the upper bounds of the best-case values, respectively.
  Moreover, we observe that $\tilde{q}_{\Lambda}^{-}$, $\tilde{q}_{\Lambda}^{+}$, and the worst-case $\Lambda$-quantiles in Table \ref{tab:1} are consistently larger than those in Table \ref{tab:2}. This shows that $\Lambda$-quantiles with  an increasing $\Lambda$ tend to penalize more in scenarios with significant capital losses. {Finally, we observe from Table \ref{tab:2} that the numerical values of the four 
$\Lambda$-quantiles are the same. In fact, their theoretical values are also the same.  This is because of the special choice of $\Lambda$ and distribution functions. In general, those four values are not the same; see e.g., Examples 4 and 5 in \cite{BP22}.}\end{example}
  
Our results  in Theorem \ref{thm:2} can be applied to robust portfolio selection. Let $\mathbf{X}=(X_1,...,X_n)\in\mathcal  X^n$ represent the loss or negative return for $n$ different assets. The set of all possible portfolio positions is denoted by $\Delta_n=\left\{\mathbf{w}=\left(w_1, \ldots, w_n\right) \in [0,1]^n: \sum_{i=1}^n w_i=1\right\}$. Note that here short-selling is not allowed. For $\mathbf{w} \in \Delta _n$, the risk of the portfolio $\mathbf{w}^{\top} \mathbf{X}$  is evaluated by $\rho\left(F_{\mathbf{w}^{\top} \mathbf{X}}\right)$, where $\rho$ is some risk measure. Suppose that we know the mean and the upper bound of the covariance matrix  of $\mathbf{X}$. Then the uncertainty set of the portfolio is given by \begin{equation}\label{eq:Mp_mv}\widehat{\M}(\mathbf{w}, \boldsymbol{\mu}, \Sigma)=\left\{F_{\mathbf{w}^{\top} \mathbf{X}}\in\mathcal M_2: \mathbb{E}[\mathbf{X}]=\boldsymbol{\mu}, \operatorname{Cov}(\mathbf{X})\preceq \Sigma\right\},
\end{equation}
where $\Sigma$ is a semidefinite symmetric matrix and for a semidefinite symmetric matrix $B$, $B\preceq \Sigma$ means $\Sigma-B$ is positive semidefinite. It follows from the result of \cite{P07} that $\widehat{\M}(\mathbf{w}, \boldsymbol{\mu}, \Sigma)=\M_2(\mathbf{w}^\top\boldsymbol{\mu}, \sqrt{\mathbf{w}^\top\Sigma \mathbf{w}})$. The optimal portfolio selection with known mean and upper bound of the covariance matrix is formulated as follows
$$ \min _{\mathbf{w} \in {\Delta_n}} \sup_{F\in\widehat{\M}(\mathbf{w}, \boldsymbol{\mu}, \Sigma)} \rho\left(F\right).
$$
We refer to \cite{PP18} for an overview of the portfolio selection with model uncertainty.
Note that return of the portfolio can also be considered as a constraint in the above optimization problem.  For instance, we can impose the constraint $\sum_{i=1}^n w_i\mu_i\leq c$ for some $c<0$ on the above optimization problem, requiring that the expected return of the portfolio is larger than $-c$. Then the optimization problem becomes $\min _{\mathbf{w} \in {\Delta_n'}} \sup_{F\in\widehat{\M}(\mathbf{w}, \boldsymbol{\mu}, \Sigma)} \rho\left(F\right)$ with $\Delta_n'=\left\{\mathbf{w}\in \Delta_n: \sum_{i=1}^n w_i\mu_i\leq c\right\}$, which does not change the nature of the problem. To simplify the problem, we here consider the portfolio selection problem with $c\geq \max_{i=1}^n \mu_i$.
 \begin{proposition}\label{prop:4}
Let $\widehat{\M}(\mathbf{w}, \boldsymbol{\mu}, \Sigma)$  be given  in  \eqref{eq:Mp_mv} and $\Lambda:\R\to [0,1]$.   For  $\rho=\tilde{q}_{\Lambda}^+$ with decreasing $\Lambda$ or $\rho=\tilde{q}_{\Lambda}^-$, we have
\begin{align}\label{mean-variance} \min _{\mathbf{w} \in {\Delta_n}} \sup_{F\in \widehat{\M}(\mathbf{w}, \boldsymbol{\mu}, \Sigma)} \rho\left(F\right)=\min _{\mathbf{w} \in {\Delta_n}} \rho\left(F^-_{\mathcal{M}_2(\mathbf{w}^{\top} \boldsymbol{\mu}, \sqrt{\mathbf{w}^{\top} \Sigma \mathbf{w}})}\right),
\end{align} where   $$  F^-_{\widehat{\M}(\mathbf{w}, \boldsymbol{\mu}, \Sigma)}(x)=\frac{(x-\mathbf{w}^{\top} \boldsymbol{\mu})^2}{(\mathbf{w}^{\top} \Sigma \mathbf{w}+(x-\mathbf{w}^{\top} \boldsymbol{\mu})^2)} \id_{\{x\geq \mathbf{w}^{\top} \boldsymbol{\mu}\}}.$$
Moreover, the optimal portfolio positions are given by the minimizer of the right-hand side of \eqref{mean-variance}.
 \end{proposition}
 \begin{proof}According to the general projection property in \cite{P07},   the two sets $\mathcal{M}_2(\mathbf{w}, \boldsymbol{\mu}, \Sigma)$ and $\mathcal{M}_2(\mathbf{w}^{\top} \boldsymbol{\mu}, \sqrt{\mathbf{w}^{\top} \Sigma \mathbf{w}})$ are identical. By Theorem \ref{thm:2}, we obtain  $$\sup_{F\in \widehat{\M}(\mathbf{w}, \boldsymbol{\mu}, \Sigma)} \rho\left(F\right)= \rho\left(F^-_{\mathcal{M}_2(\mathbf{w}^{\top} \boldsymbol{\mu}, \sqrt{\mathbf{w}^{\top} \Sigma \mathbf{w}})}\right),$$
 which further implies
 $$ \min _{\mathbf{w} \in {\Delta_n}} \sup_{F\in \widehat{\M}(\mathbf{w}, \boldsymbol{\mu}, \Sigma)} \rho\left(F\right)=\min _{\mathbf{w} \in {\Delta_n}} \rho\left(F^-_{\mathcal{M}_2(\mathbf{w}^{\top} \boldsymbol{\mu}, \sqrt{\mathbf{w}^{\top} \Sigma \mathbf{w}})}\right).$$
We complete the proof.
 \end{proof}

 {\begin{example}\label{exm:5} We  consider the case of $n=2$. Let   $p=2$ and $\Lambda_{(\alpha,\beta;z)}: x\mapsto \alpha \id_{\{x< z\}} + \beta \id_{\{x\geq z\}}$. Assume that $\boldsymbol{\mu}_1=(0.5,1)^\top$, $\boldsymbol{\mu}_2=(5,6)^\top$, and  the  covariance matrices are  given by  \begin{equation}\label{eq:matric}\Sigma_1=\left(\begin{array}{cc}1 & 0.5 \\ 0.5 & 1\end{array}\right),~~\text{and}~~\Sigma_2=\left(\begin{array}{cc}1 & -0.5 \\- 0.5 & 1\end{array}\right).\end{equation}
Here, we consider the worst-case values  of $\tilde{q}^-_{\Lambda}$, as $\sup_{F \in \M_p(m,v)} \tilde{q}_{\Lambda}^-(F) = \tilde{q}_{\Lambda}^-(F_{\M_p(m,v)}^-)$ holds for both increasing and decreasing $\Lambda$. As a comparison, we also consider the worst-case values of $q_{\alpha}^-$ for some $\alpha\in (0,1)$

Figures \ref{fig:1} and \ref{fig:2} demonstrate the worst-case values of  the portfolio as functions of $w_1$, showing that the optimal portfolio positions under the $\Lambda$-quantile criterion can differ from those under the VaR criterion. Additionally, for an increasing $\Lambda$ in Figure \ref{fig:1}, when the mean of the portfolio is relatively small, the  DM uses a small probability level of $0.8$ to determine the portfolio positions; when the mean of the portfolio is relatively large, a larger  probability  level of $ 0.95$ is used. However, if  $\Lambda$ is a decreasing function in Figure \ref{fig:2}, this situation is reversed. This phenomenon is due to the fact that  an increasing $\Lambda$ may penalize large losses, while a decreasing $\Lambda$ suggests that the DM accepts a higher probability for larger potential losses.  Moreover, when two
assets are positively correlated, the worst-case values are higher compared to the case with
negatively correlated assets, as a negative correlation results in a hedging effect. 
	 \begin{figure}[hbt!]
\centering
 \includegraphics[width=15cm]{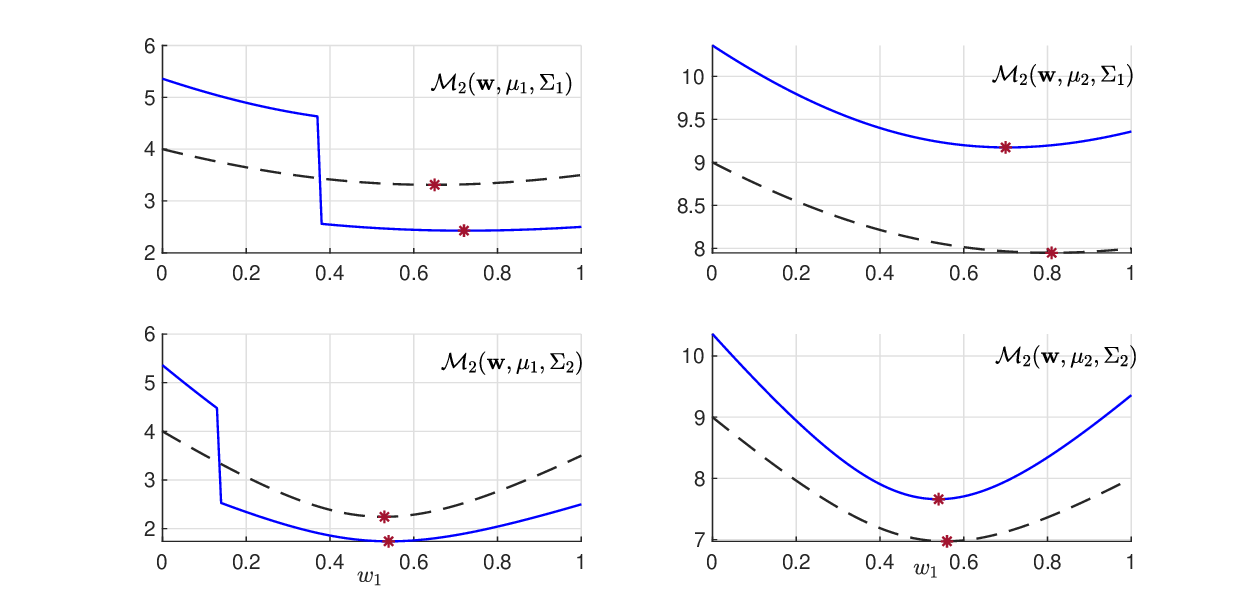}
 \caption{{ The worst-case values of $\tilde q^-_{\Lambda}$ with  $\Lambda=\Lambda_{(0.8,0.95;3)}$ (blue solid line)   and $q^-_\alpha$  with $\alpha=0.9$ (black dashed line) as functions of $w_1$}  }\label{fig:1}

\centering
 \includegraphics[width=15cm]{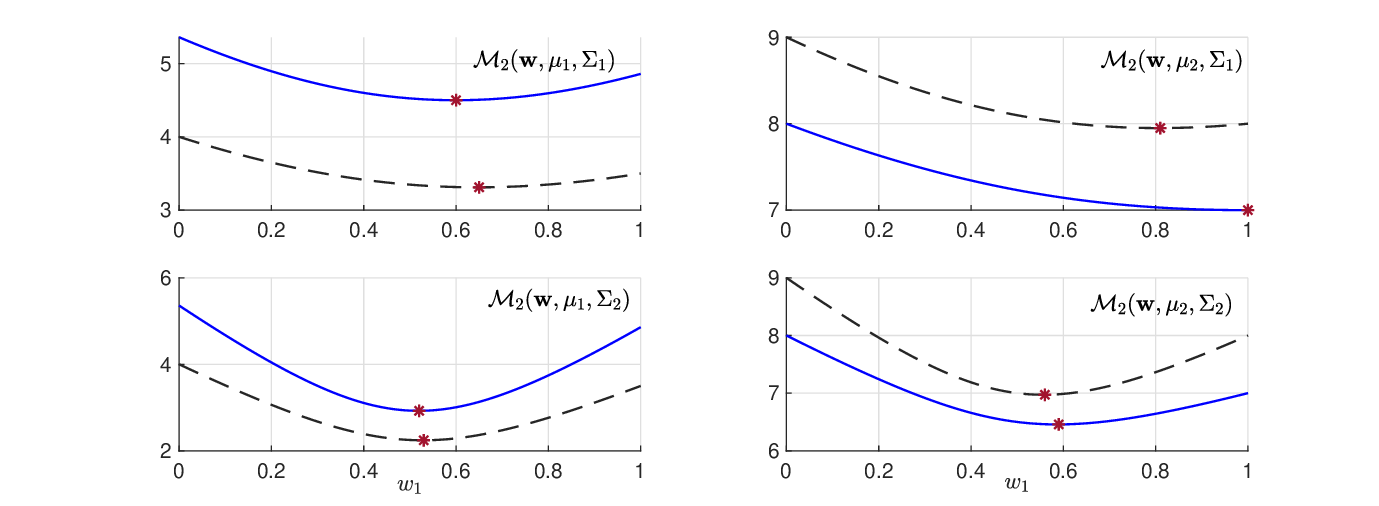}
 \caption{{ The worst-case values of $\tilde q^-_{\Lambda}$  with  $\Lambda=\Lambda_{(0.95,0.8;6)}$ (blue solid line)  and $q^-_\alpha$ with $\alpha=0.9$ (black dashed line) as functions of $w_1$}}\label{fig:2}
\end{figure}
		\end{example}
	
\subsection{Probability constraints via Wasserstein distance} \label{sec:4.2}
The Wasserstein metric is a popular notion used in mass transportation and distributionally robust optimization; see e.g., \cite{EK18} and \cite{BM19}. In the one-dimensional setting, the Wasserstein metric has an explicit formula.   For $p \geqslant 1$ and $F, G \in \mathcal{M}_p$, the $p$-Wasserstein distance between $F$ and $G$ is defined as
\begin{equation}\label{eq:W_p}
W_p(F, G)=\left(\int_0^1\left|F^{-1}(u)-G^{-1}(u)\right|^p \mathrm{~d} u\right)^{1 / p}.
\end{equation}
For $\varepsilon \geqslant 0$, the uncertainty set of an $\varepsilon$-Wasserstein ball around a baseline distribution $G \in \mathcal{M}_p$ is given by
\begin{equation}\label{def:M_ep}
\mathcal{M}_p({G}, \varepsilon)=\left\{F \in \mathcal{M}_p: W_p(F, {G}) \leqslant \varepsilon\right\},
\end{equation}
where the parameter $\epsilon$ represents the magnitude of uncertainty. 
  In the following result, we will give the expressions of  $F_{\mathcal{M}_p(G,\varepsilon)}^-$ and $F_{\mathcal{M}_p(G,\varepsilon)}^+$.
\begin{proposition}
\label{lemma:2}  For $\varepsilon>0, p \geqslant 1$ and $G \in \mathcal{M}_p$, we have
\begin{align}\label{eq:eps1}
F_{\M_p(G,\epsilon)}^-(x)=l_+^{-1}(x),~F_{\M_p(G,\epsilon)}^+(x)=u_+^{-1}(x),~x\in\R,
\end{align}
where for {each} $\alpha\in (0,1)$, $l(\alpha)$ is the unique solution of the following equation
$$
\int_\alpha^1\left(l(\alpha)-q^-_t (G)\right)_{+}^p \mathrm{~d} t=\epsilon^p,
$$
and $u(\alpha)$ is the unique solution of the following equation
\begin{equation*}
\int_0^\alpha\left(q_t^-(G)-u(\alpha)\right)_{+}^p \mathrm{~d} t=\varepsilon^p.
\end{equation*}
Moreover,  both $\widehat{F}_{\M_p(G,\epsilon)}^-$ and $F_{\M_p(G,\epsilon)}^+$ are continuous and attainable.
\end{proposition}
\begin{proof} By (iii) of Proposition 4 of \cite{LMWW22}, we have for $\alpha\in (0,1)$,  $ \sup_{F\in \mathcal{M}_p(G, \varepsilon)} q^-_\alpha(F) =l(\alpha)$ is the unique solution of the following equation
$$
\int_\alpha^1\left(l(\alpha)-q^-_t (G)\right)_{+}^p \mathrm{~d} t=\epsilon^p.
$$ By (i) of Corollary \ref{coro:1}, we have $\sup_{F\in\M_p(G,\epsilon)}q_{\alpha}^-(F)=q_{\alpha}^-(F_{\M_p(G,\epsilon)}^-)=l(\alpha)$. For $0<\alpha_1<\alpha_2<1$, it follows that $\int_{\alpha_2}^1\left(l(\alpha_1)-q^-_t (G)\right)_{+}^p \mathrm{~d} t<\epsilon^p$. This implies that $l(\alpha_2)>l(\alpha_1)$.
 Hence, $l$ is strictly increasing over $(0,1)$, which implies that $F_{\M_p(G,\epsilon)}^-(x)=l_+^{-1}(x)$ for $x\in\R$ and $F_{\M_p(G,\epsilon)}^-$ is continuous.

Next, we show the second statement.
  Note that $\inf_{F\in \M_p(G,\epsilon)}q_{\alpha}^+(F)=-\sup_{F\in \M_p(\overline{G},\epsilon)}q_{1-\alpha}^-(F):=u(\alpha)$, where $\overline{G}^{-1}(t)=-G^{-1}(1-t),~t\in (0,1)$. Using the above conclusion, we have $-u(\alpha)$ is the unique solution of
 $$
\int_{1-\alpha}^1\left(-u(\alpha)+q^-_{1-t} (G)\right)_{+}^p \mathrm{~d} t=\epsilon^p.
$$
Hence, $u(\alpha)$ is the unique solution of
$$
\int_0^\alpha\left(q^-_{t} (G)-u(\alpha)\right)_{+}^p \mathrm{~d} t=\epsilon^p.
$$
It follows from  (i) of Corollary \ref{coro:1} that $q_{\alpha}^+(F_{\M_p(G,\epsilon)}^+)=\inf_{F\in \M_p(G,\epsilon)}q_{\alpha}^+(F)=u(\alpha)$.
Moreover, for $\alpha_1<\alpha_2$, it follows that $\int_0^{\alpha_2}\left(q^-_{t} (G)-u(\alpha_1)\right)_{+}^p \mathrm{~d} t>\epsilon^p$, which implies $u(\alpha_2)>u(\alpha_1)$. Hence, $u$ is strictly increasing on $(0,1)$. Consequently, we have $F_{\M_p(G,\epsilon)}^+(x)=u_+^{-1}(x),~x\in\R$ and $F_{\M_p(G,\epsilon)}^+$ is continuous.

Next, we show the attainability of $\widehat{F}_{\M_p(G,\epsilon)}^-$ and $F_{\M_p(G,\epsilon)}^+$.
 For any $x\in\R$, there exists a sequence $F_k\in \M_p(G,\epsilon)$ such that
$\lim_{k\to\infty} F_k(x-)=\widehat{F}_{\M_p(G,\epsilon)}^-(x)$. It follows that
$$\int_{0}^{1}|F_k^{-1}(t)|^p\d t\leq \left(W_p(\delta_0,G)+W_p(F_k, G)\right)^{p}\leq \left(W_p(\delta_0,G)+\epsilon\right)^{p},$$
where $\delta_x$ is the distribution with probability mass 1 on $x$.
By Chebyshev's inequality, we have $\sup_{k\geq 1}(1-F_k(x))\leq \frac{\left(W_p(\delta_0,G)+\epsilon\right)^{p}}{x^p}$ for $x>0$ and $\sup_{k\geq 1}F_k(x)\leq \frac{\left(W_p(\delta_0,G)+\epsilon\right)^{p}}{|x|^p}$ for $x<0$. Hence, $\{F_{k},~k\geq 1\}$ is tight. By Theorem 25.10 of \cite{B95}), there exist a subsequence $\{F_{k_l},~l\geq 1\}$ and a distribution $F$ such that $F_{k_l}\to F$ weakly as $l\to\infty$, which implies $F_{k_l}^{-1}(t) \to F^{-1}(t)$ over $(0,1)$ a.s.. It follows from Fatou's lemma that
$$ \int_{0}^{1}|F^{-1}(t)-G^{-1}(t)|^p\d t=\int_{0}^{1}\liminf_{l\to\infty}|F_{k_l}^{-1}(t)-G^{-1}(t)|^p\d t\leq \liminf_{l\to\infty}\int_{0}^{1}|F_{k_l}^{-1}(t)-G^{-1}(t)|^p\d t\leqslant \epsilon^p.$$
Hence, $F\in \M_p(G,\epsilon)$.  Using the same argument as in the proof of Proposition \ref{lemma:1}, we obtain $F(x-)=\widehat{F}_{\M_p(G,\epsilon)}^-(x)$. Hence, $\widehat{F}_{\M_p(G,\epsilon)}^-$ is attainable. Using the similar argument, we can show that $F_{\M_p(G,\epsilon)}^+$ is attainable. The details are omitted.
\end{proof}

Following  Theorem \ref{Th:main}, Proposition \ref{prop:attainable}, and Proposition \ref{lemma:2}, we immediately obtain the following conclusions.
 \begin{theorem}\label{thm:3} For $\varepsilon>0, p \geqslant 1$, $ G \in \mathcal{M}_p$, and $\Lambda:\R\to [0,1]$, let $F_{\mathcal{M}_p(G, \varepsilon)}^-$ and $F_{\mathcal{M}_p(G, \varepsilon)}^+$ be given by \eqref{eq:eps1}. Then we have
 \begin{enumerate}[(i)]
 \item $\sup_{F\in \mathcal{M}_p(G, \varepsilon)}\tilde{q}_{\Lambda}^-(F)=\tilde{q}_{\Lambda}^-(F_{\mathcal{M}_p(G, \varepsilon)}^-)$; $\inf_{F\in \mathcal{M}_p(G, \varepsilon)}q_{\Lambda}^+(F)=q_{\Lambda}^+(F_{\mathcal{M}_p(G, \varepsilon)}^+)$;  $\inf_{F\in\mathcal{M}_p(G, \varepsilon)}q_{\Lambda}^-(F)=q_{\Lambda}^-(F_{\mathcal{M}_p(G, \varepsilon)}^+)$;
\item  If $\Lambda$ is decreasing, then (i) remains true by replacing $\tilde{q}_\Lambda^-$ by $q_\Lambda^-$, $q_\Lambda^+$  by $\tilde{q}_\Lambda^+$ and  $q_\Lambda^-$ by $\tilde{q}_\Lambda^-$; Moreover, $\sup_{F\in\mathcal{M}_p(G, \varepsilon)}\tilde{q}_{\Lambda}^+(F)=\sup_{F\in\mathcal{M}_p(G, \varepsilon)}q_{\Lambda}^+(F)=\tilde{q}_{\Lambda}^+(F_{\mathcal{M}_p(G, \varepsilon)}^-)$.
 \end{enumerate}
  \end{theorem}
 In the following example, we  consider the uncertainty set $\mathcal{M}_p(G,\varepsilon)$ with $G$ being  normal, exponential and uniform distributions. Note that $F_{ \M_p(G,\epsilon)}^-$  and  $F_{ \M_p(G,\epsilon)}^+$  can be computed numerically via \eqref{eq:eps1}.
   For decreasing $\Lambda$,  it is more convenient to calculate  the best-case and worst case of $\Lambda$ -quantiles using the expressions in  Proposition \ref{prop:2} as we do not need to compute the inverse functions of $l$ and $u$.
 \begin{example}\label{exm:6} Let  $p=1$, $\epsilon=0.1$ and  $\Lambda_{(\alpha,\beta;z)}: x\mapsto \alpha \id_{\{x<z\}} + \beta \id_{\{x\geq z\}}$.    Applying the results in Theorem \ref{thm:3} and Proposition \ref{prop:2}, we obtain the robust $\Lambda$-quantiles numerically, which are displayed in  Tables \ref{tab:3} and \ref{tab:4}.   Again,   the bold numbers in Table \ref{tab:3}  represent the upper bounds for the worst-case values and the lower bounds for the best-case values.

{We observe that  the values and bounds of the four $\Lambda$-quantiles displayed in Tables \ref{tab:3} and \ref{tab:4} are all the same.} Moreover, the worst-case values of $\tilde{q}_{\Lambda}^{-}$ with an increasing $\Lambda$ are larger than those with a decreasing $\Lambda$. This finding indicates that $\tilde{q}_{\Lambda}^{-}$ with increasing $\Lambda$ functions are more conservative under this uncertainty set.
\begin{table}[h]
\def\arraystretch{1} \begin{center}  \caption{Robust  $\Lambda$-quantiles   for an increasing $\Lambda$ with  $\mathcal M=\mathcal{M}_p(G, \varepsilon)$}   \label{tab:3} \begin{tabular}{c|cccccc}
 $\Lambda_{(0.8,0.95;2)}/ \text{best-case}$ & $ q_{\Lambda}^{-}  $ &   $q_{\Lambda}^{+} $&$ \tilde{q}_{\Lambda}^{-}  $&$ \tilde{q}_{\Lambda}^{+}$   \\ \hline
 $ \mathrm{N}(1,1)$ &  1.90  &    1.90  &    \textbf{1.90}& \textbf{1.90}  \\
   $  \mathrm{exp}(1)$    &2.00&2.00&\textbf{2.00}& \textbf{2.00}  \\ $U[1-\sqrt{3},1+\sqrt{3} ]$  &1.89  &1.89 & \textbf{1.89}& \textbf{1.89}  \\  \hline
   $\Lambda_{(0.8,0.95;2})/ \text{worst-case}$ & $ q_{\Lambda}^{-}  $ &   $q_{\Lambda}^{+} $&$ \tilde{q}_{\Lambda}^{-}  $&$ \tilde{q}_{\Lambda}^{+}$   \\ \hline
 $\mathrm{N}(1,1)$ & \textbf{5.06}&    \textbf{5.06}  &
5.06 &  \textbf{5.06}  \\
   $  \mathrm{exp}(1)$    &\textbf{6.00}&\textbf{6.00}&6.00& \textbf{6.00}\\
   $ U[1-\sqrt{3},1+\sqrt{3} ]$  &\textbf{4.65}  &\textbf{4.65} & 4.65& \textbf{4.65}   \\
    \hline
   \end{tabular}  \end{center}\end{table}
 \begin{table}[h]
\def\arraystretch{1} \begin{center}  \caption{ Robust  $\Lambda$-quantiles   for a decreasing  $\Lambda$ with  $\mathcal M=\mathcal{M}_p(G, \varepsilon)$}  \label{tab:4}  \begin{tabular}{c|cccccc}   $\Lambda_{(0.95,0.8;2)}/ \text{best-case}$ & $ q_{\Lambda}^{-}  $ &   $q_{\Lambda}^{+} $&$ \tilde{q}_{\Lambda}^{-}  $&$ \tilde{q}_{\Lambda}^{+}$   \\ \hline
 $ \mathrm{N}(1,1)$ &    1.06&  1.06 &     1.06   &    1.06    \\
   $ \mathrm{exp}(1)$    &2.11&2.11& 2.11& 2.11 \\
   $ U[1-\sqrt{3},1+\sqrt{3} ]$  &0.65 &0.65 &0.65& 0.65   \\
   \hline  $\Lambda_{(0.95,0.8;2)}/ \text{worst-case}$ & $ q_{\Lambda}^{-}  $ &   $q_{\Lambda}^{+} $&$ \tilde{q}_{\Lambda}^{-}  $&$ \tilde{q}_{\Lambda}^{+}$   \\ \hline
 $\mathrm{N}(1,1)$ &   2.90  &  2.90  &  2.90  &  2.90   \\
   $ \mathrm{exp}(1)$    &3.11&3.11& 3.11& 3.11 \\ $ U[1-\sqrt{3},1+\sqrt{3} ]$  &2.88  &2.88  & 2.88 & 2.88    \\\hline
 \end{tabular}  \end{center}\end{table}

\end{example}

We next focus on the portfolio selection problem.  For $p \geqslant 1$ and $a \geqslant 1$, the $p$-Wasserstein metric on $\mathbb{R}^n$ between $F, G \in \mathcal{M}_p\left(\mathbb{R}^n\right)$ is defined as
$$
W_{a, p}^n(F, G)=\inf _{F_ \mathbf{X}=F, F_ \mathbf{Y}=G}\left(\mathbb E\left[\|\mathbf{X}-\mathbf{Y}\|_a^p\right]\right)^{1 / p},
$$
where   $F_{\mathbf{X}}$ is  the cdf  of $\mathbf{X}$,  and  $\|\cdot\|_a$ is the $L^a$ norm on $\mathbb{R}^n$.  If $n=1$, then $W_{a, p}^d$ is $W_p$ in \eqref{eq:W_p}  where the infimum is attained by comonotonicity via the Fr\'echet-Hoeffing inequality. Define the Wasserstein uncertainty set for a benchmark distribution $G \in \mathcal{M}_p(\mathbb{R}^n)$ as
$$
\mathcal{M}_{a, p}^n(G,\epsilon)=\left\{F \in \mathcal{M}_p\left(\mathbb{R}^n\right): W_{a, p}^n(F, G) \leqslant \varepsilon\right\}, ~\varepsilon >0.
$$
 The univariate uncertainty set for the cdf of $\mathbf{w}^{\top} \mathbf{X}$ is denoted by
\begin{equation*}\label{eq:M_a_p}
\mathcal{M}_{\mathbf{w}, a, p}(G,\varepsilon)=\left\{F_{\mathbf{w}^{\top} \mathbf{X}}: F_{\mathbf{X}} \in \mathcal{M}_{a, p}^n(G,\epsilon)\right\}, \quad G \in \mathcal{M}_p(\mathbb{R}^d).
\end{equation*}
  We next solve the following robust portfolio selection problem
\begin{equation}\label{eq:opt2}
 \min _{\mathbf{w} \in {\Delta_n}} \sup_{F\in\mathcal{M}_{\mathbf{w}, a, p}(G,\varepsilon)} \rho\left(F\right).
\end{equation}
\begin{lemma}[Theorem 5 of \cite{MWW22}]\label{lem:3} For $\varepsilon \geqslant 0, a \geqslant 1$ and $p \geqslant 1$, random vector $\mathbf{X}$ with $F_{\mathbf{X}} \in \mathcal{M}_p\left(\mathbb{R}^n\right)$ and $\mathbf{w} \in \mathbb{R}^n$, we have
$$
\mathcal{M}_{\mathbf{w}, a, p}\left(F_{\mathbf{X}},\epsilon\right)=\mathcal{M}_{p}\left(F_{\mathbf{w}^{\top} \mathbf{X}}, \epsilon\|\mathbf{w}\|_b\right),
$$
where $b$ satisfies $1 / a+1 / b=1$.
\end{lemma}
Based on Lemma \ref{lem:3},   the  optimization problem \eqref{eq:opt2} can be solved using Theorem \ref{thm:3}.
 \begin{proposition}\label{prop:5}Suppose that $\varepsilon>0$,  $p \geqslant 1$, $a \geqslant 1$, a random vector $\mathbf{X}$ satisfying $F_{\mathbf{X}} \in \mathcal{M}_p\left(\mathbb{R}^n\right)$ and $\Lambda:\R\to [0,1]$.  For $\rho=\tilde{q}_{\Lambda}^+$ with decreasing $\Lambda$ or $\rho=\tilde{q}_{\Lambda}^-$,  we have
  \begin{align}\label{eq:selection1}\min _{\mathbf{w} \in {\Delta_n}} \sup_{F\in\mathcal{M}_{\mathbf{w}, a, p}(G,\varepsilon)}\rho\left(F\right)=\min _{\mathbf{w} \in {\Delta_n}} \rho\left(F^-_{\mathcal{M}_{p}\left(F_{\mathbf{w}^{\top} \mathbf{X}}, \epsilon\|\mathbf{w}\|_{a/(a-1)}\right)}\right),
  \end{align}
 where
 $F^-_{\mathcal{M}_{p}\left(F_{\mathbf{w}^{\top} \mathbf{X}}, \epsilon\|\mathbf{w}\|_{a/(a-1)}\right)}(x)=u_+^{-1}(x),~x\in\R$ with $u(\alpha)$ being the unique solution of the following equation
$$\int_\alpha^1\left(u(\alpha)-F_{\mathbf{w}^\top\mathbf{X}}^{-1}(s)\right)_{+}^p \mathrm{~d} s=(\epsilon\|\mathbf{w}\|_{a/(a-1)})^p, \quad \alpha \in(0,1).$$
Moreover, the optimal portfolio positions are given by the minimizer of the right-hand side of \eqref{eq:selection1}.
 \end{proposition}
\begin{proof} Using Theorem \ref{thm:3} and Lemma \ref{lem:3}, we obtain the results immediately.
\end{proof}

\begin{example} Let the benchmark distribution $F_0=\textbf{t}(\tau,\boldsymbol{\mu}_i, \Sigma_i)$ ($i=1,2$), and denote by $F_\tau=\textbf{t} (\tau,0,1)$. Then for $x\in\R$,  we have $F^-_{\mathcal{M}_{p}\left(F_{\mathbf{w}^{\top} \mathbf{X}}, \epsilon\|\mathbf{w}\|_{a/(a-1)}\right)}(x)=u_+^{-1}(x)$ with $u(\alpha)$ being the unique solution of the following equation  $$ \int_\alpha^1\left(u(\alpha)-\mathbf{w}^{\top} \boldsymbol{\mu}_i-\sqrt{\mathbf{w}^{\top} \Sigma_i \mathbf{w}} F_\tau^{-1}(s)\right)^p_+\d s= (\epsilon\|\mathbf{w}\|_{a/(a-1)})^p.$$  We consider the case of   $n=2$. Let  $p=1$,  $a=2$, $\epsilon=0.1$, $\tau=3$ and $\Lambda_{(\alpha,\beta;z)}: x\mapsto \alpha \id_{\{x< z\}} + \beta \id_{\{x\geq z\}}$.  We assume that $\boldsymbol{\mu}_1=(0.5,1)^\top$, $\boldsymbol{\mu}_2=(3,4)^\top$, and  the  covariance matrices are  given by \eqref{eq:matric}.

 Figures \ref{fig:3} and \ref{fig:4}   show  the worst-case values of the   portfolio { as functions of  $w_1$}.   The optimal portfolio positions under the {$\Lambda$-quantile}  criterion can differ from those under the VaR criterion.   In Figure \ref{fig:3}, with  an increasing $\Lambda$, when the mean of the benchmark distribution is relatively small, the DM uses the small probability level of $0.8$ to determine the portfolio positions; when the mean of the portfolio is relatively large, a larger  probability  level of $0.95$ is used. But in Figure \ref{fig:4},  where  $\Lambda$ is a decreasing function, this situation is reversed.  Additionally, the optimal portfolio positions are higher when the covariance matrix of the  benchmark distribution is  positively correlated compared to the case with negatively correlated assets. All these phenomena can be explained similarly to Example \ref{exm:5}.

\begin{figure}[h!]
\centering
 \includegraphics[width=15cm]{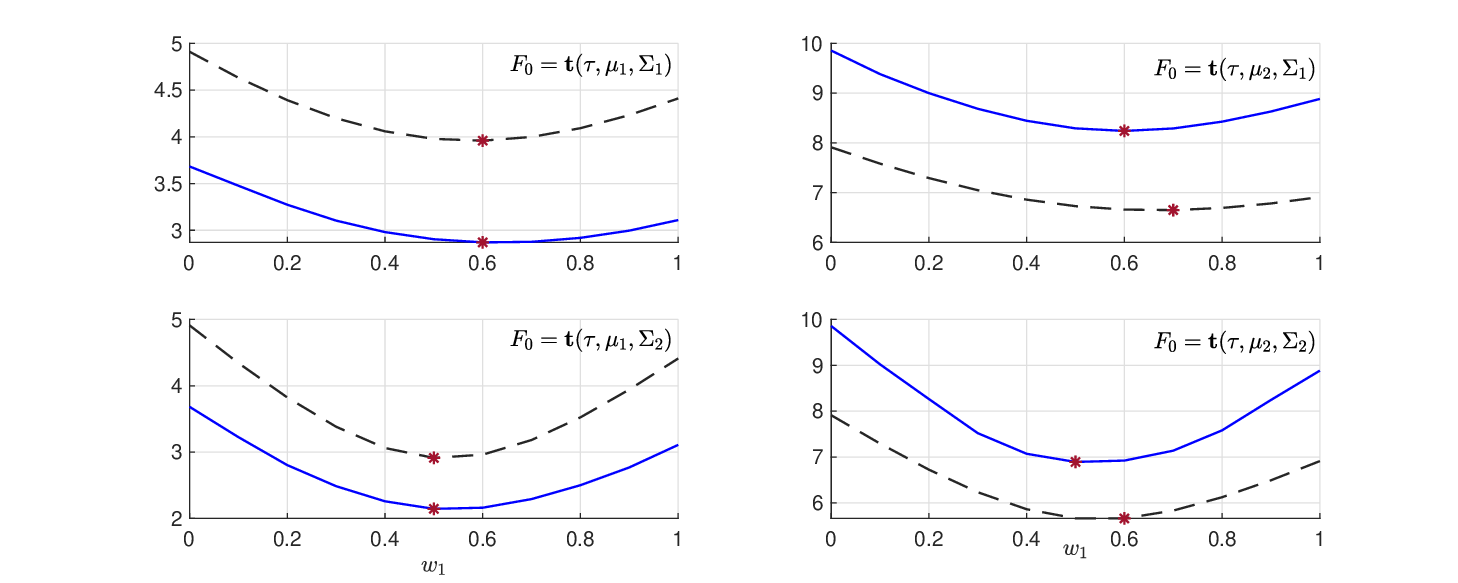}
 \caption{{ The worst-case values of $\tilde q^-_{\Lambda}$ with $\Lambda=\Lambda_{(0.8,0.95;6)}$ (blue solid line)   and  $q^-_\alpha$   with $\alpha=0.9$ (black dashed line) as functions of $w_1$}}\label{fig:3}
  \includegraphics[width=15cm]{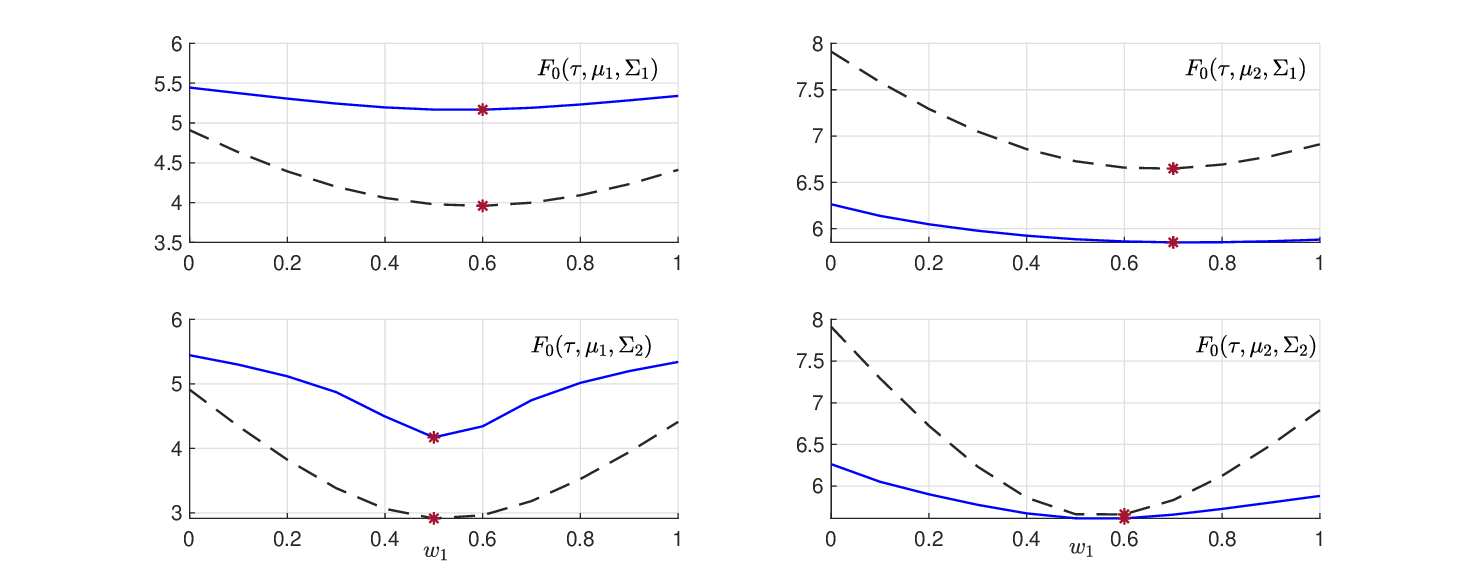}
   \caption{{ The worst-case values of $\tilde q^-_{\Lambda}$ with   $\Lambda=\Lambda_{(0.95,0.8;6)}$ (blue solid line)  and $q^-_\alpha$ with $\alpha=0.9$ (black dashed line) as functions of $w_1$}}\label{fig:4}
\end{figure}

 \end{example}

\subsection{Marginal constraints in risk aggregation}\label{Sec:4.3}  For a loss vector $(X_1,\dots, X_n)$, suppose that the marginal distribution $X_i\sim F_i$ is known but the dependence structure is completely unknown. This assumption is motivated from the context where data from different correlated products are
separately collected and thus no dependence information is available; see \cite{EPR13} and \cite{EWW15}. All the possible distributions of the total loss $X_1+\dots+X_n$ are characterized by the aggregation set as below $$ \mathcal{D}_n(\mathbf{F})=\left\{\text { cdf of } X_1+\cdots+X_n: X_i \sim F_i, i=1, \ldots, n\right\} ,$$
where $\mathbf{F}=\left(F_1, \ldots, F_n\right) \in \mathcal{M}^n$.
 The worst-case and best-case values of $\rho$ over this uncertainty set are given by $$ \overline{\rho}(\mathbf{F})=\sup \left\{\rho(G): G \in \mathcal{D}_n(\mathbf{F})\right\}, ~~\text{and}~~\underline{\rho}(\mathbf{F})=\inf \left\{\rho(G): G \in \mathcal{D}_n(\mathbf{F})\right\}.$$
 The worst-case value of quantiles over the aggregation set has been studied extensively in the literature.  It is well known that the  worst-case value of  quantiles in risk aggregation generally does not admit any analytical formula.  Typically, the results in the literature either offer explicit expressions for some special type of marginal distributions or provide some upper bounds for the general marginal distributions.
We refer to \cite{EP06} for bounds, and \cite{WW16}, \cite{JHW16} and \cite{BLLW20} for the analytical expressions.

{ We first offer a variant version of  Theorem \ref{Th:main} and Proposition \ref{prop:attainable}, which would be useful to our results in this subsection.
For $\Lambda:\R\to [0,1]$, let $\alpha_{\Lambda}^-=\inf_{x\in\R}\Lambda(x)$ and $\alpha_{\Lambda}^+=\sup_{x\in\R}\Lambda(x)$. The following observation follows directly from the definition of $\Lambda$-quantiles.
\begin{lemma}\label{lem:1}
    For an increasing function $f:\R\to [0,1]$ and $\rho \in \{q_{\Lambda}^-, q_{\Lambda}^+, \tilde{q}_{\Lambda}^-, \tilde{q}_{\Lambda}^+\}$, if $\beta_1<\alpha_{\Lambda}^-$ and $\beta_2>\alpha_{\Lambda}^+$, then we have 
$\rho(f)=\rho(\beta_1\vee f)=\rho(\beta_2\wedge f)$. 
\end{lemma}

For $t\in\R$, we say that $t\vee F_{\M}^-$ (resp. $t\wedge F_{\M}^+$) is \emph{attainable}  if for any $z\in\R$, there exists $F\in\M$ such that $t\vee F(z)=t\vee F_{\M}^-(z)$ (resp. $t\wedge F(z)=t\wedge F_{\M}^+(z)$).
Based on the above lemma, we have the following version of Theorem \ref{Th:main} and Proposition \ref{prop:attainable}. 
\begin{proposition}\label{prop:new} For  $\beta_1<\alpha_{\Lambda}^-$ and $\beta_2>\alpha_{\Lambda}^+$,  Theorem \ref{Th:main} and Proposition \ref{prop:attainable} are still valid if $F_{\M}^-$ is replaced by $\beta_1\vee F_{\M}^-$, $F_{\M}^+$ by $\beta_2\wedge F_{\M}^+$ and $\widehat{F}_{\M}^-$ by $\beta_1\vee \widehat{F}_{\M}^-$.
\end{proposition}
\begin{proof}
By Lemma \ref{lem:1}, we immediately obtain that (i) of Theorem \ref{Th:main} still holds when $F_{\M}^-$ is replaced  by  $\beta_1\vee F_{\M}^-$ and $F_{\M}^+$ by $\beta_2\wedge F_{\M}^+$.   Furthermore, using Lemma \ref{lem:1} and the attainability of $\beta_1\vee F_{\M}^-$ or $\beta_2\wedge F_{\M}^+$, we can show  that (ii) of Theorem \ref{Th:main} still holds after the replacement. The proof is almost the same as before and hence it is omitted.  With the replacement in place,  (iii) of Theorem \ref{Th:main} follows directly from (i) and (ii) of Theorem \ref{Th:main},  along with the fact that $\Lambda$ is decreasing. 
Finally, the proof of Proposition \ref{prop:attainable} is almost the same as before and hence, it is omitted.
\end{proof}}

{We say that  a distribution $F$ has a monotone density beyond its $t$-quantile for $t\in [0,1)$ if $\frac{(F(x)-t)_+}{1-t}$ has a monotone density on its support. Similarly, we say that a distribution $F$ has a monotone density below its $t$-quantile for $t\in (0,1]$ if $\frac{F(x)\wedge t}{t}$ has a monotone density on its support.
 Let $\M_{t,D}^+$ (resp. $\M_{t,I}^+$) denote the set of all distributions with decreasing (resp. increasing) densities beyond their corresponding $t$-quantiles, and $\M_{t,D}^-$ (resp. $\M_{t,I}^-$) denote the set of all distributions with decreasing (resp. increasing) densities below their corresponding $t$-quantiles.}
 The expressions of the {extremal distributions} over the  uncertainty set $\mathcal{D}_n(\mathbf{F})$ are displayed in the following proposition.
\begin{proposition}\label{prop:aggregationF} Suppose that  $F_1^{-1},\cdots, F_n^{-1}$ are continuous over $(0,1)$.
For $\mathbf F\in \left(\M_{t,D}^+\right)^n \cup \left(\M_{t,I}^+\right)^n$ with $t\in [0,1)$, we have
 \begin{align}\label{eq:aggregation}
 t\vee F_{\mathcal{D}_n(\mathbf{F})}^-=t\vee (1-H_{\mathbf F}),
 \end{align}
 where $$H_{\mathbf F}(x)=\inf_{\mathbf r\in \Theta_n(x)}\left\{\sum_{i=1}^{n}\frac{1}{x-r}\int^{x-r+r_i}_{r_i}(1-F_i(y))\d y\right\}$$
 with $\mathbf r=(r_1,\dots, r_n)$, $r=\sum_{i=1}^{n}r_i$ and $
  \Theta_n = \{(r_1,\dots,r_n) \in \mathbb R^n: \sum_{i=1}^{n}r_i< x  \}$.
  
  For  $\mathbf F\in \left(\M_{t,D}^-\right)^n \cup \left(\M_{t,I}^-\right)^n$ with $t\in (0,1]$, we have
 \begin{align}\label{eq:aggregationG}
t\wedge F_{\mathcal{D}_n(\mathbf{F})}^+(x)=t\wedge H_{\overline{\mathbf F}},
 \end{align}
 where $\overline{\mathbf{F}}=(\overline{F}_1,\dots, \overline{F}_n)$ with $\overline{F}_i(x)=\lim_{y\downarrow x}1-F_i(-y),~x\in\R$.
 
  Moreover, both of $t\vee\widehat{F}_{\mathcal{D}_n(\mathbf{F})}^-$ and $t\wedge F_{\mathcal{D}_n(\mathbf{F})}^+$ are continuous and attainable.
\end{proposition}
\begin{proof}
First, we suppose that $t\in (0,1)$. In light of Theorems 2 and 4 of \cite{BLLW20}, we have
\begin{align}\label{eq:tech1}\sup_{F\in \mathcal{D}_n(\mathbf{F})}q_\alpha^{+}(F)=H_{\mathbf F}^{-1}(1-\alpha),~\alpha\in [t,1),
\end{align}
where $H_{\mathbf F}^{-1}(\alpha)=\inf\{x\in\R: H_{\mathbf F}(x)<\alpha\},~\alpha\in (0,1)$.
Using the fact $F_1^{-1},\dots, F_n^{-1}$ are continuous over $(0,1)$, by Lemma 4.5 of \cite{BJW14}, we have
 \begin{align}\label{eq:tech2}\sup_{F\in \mathcal{D}_n(\mathbf{F})}q_\alpha^{-}(F)=\sup_{F\in \mathcal{D}_n(\mathbf{F})}q_\alpha^{+}(F),~\alpha\in (0,1).
 \end{align}
 Moreover, in light of (i) of Corollary \ref{coro:1}, we have $\sup_{F\in \mathcal{D}_n(\mathbf{F})}q_\alpha^{-}(F)=q_{\alpha}^-(F_{\mathcal{D}_n(\mathbf{F})}^-)$ for $\alpha\in (0,1)$. Consequently, $q_{\alpha}^-(F_{\mathcal{D}_n(\mathbf{F})}^-)=H_{\mathbf F}^{-1}(1-\alpha),~\alpha\in [t,1).$  As it is stated in the proof of Theorem 4 of \cite{BLLW20}, $H_{\mathbf F}(x)$ is continuous over $\R$ and strictly decreasing over $(-\infty, \sum_{i=1}^{n}q_1^{-1}(F_i))$ and is equal to $0$ on $[\sum_{i=1}^{n}q_1^{-1}(F_i),\infty)$. Hence, we have  $F_{\mathcal{D}_n(\mathbf{F})}^-(x)=1-H_{\mathbf F}(x), ~x\geq H_{\mathbf F}^{-1}(1-t)$,  which implies that $t\vee F_{\mathcal{D}_n(\mathbf{F})}^-=t\vee (1-H_{\mathbf F})$. We can similarly show that the conclusion holds for $t=0$.

Suppose $t\in (0,1)$. By  (i) of Corollary \ref{coro:1}, we have $\inf_{F\in \mathcal{D}_n(\mathbf{F})}q_\alpha^{+}(F)=q_{\alpha}^+(F_{\mathcal{D}_n(\mathbf{F})}^+)$ for $\alpha\in (0,1)$. Direct computation gives $\inf_{F\in \mathcal{D}_n(\mathbf{F})}q_\alpha^{+}(F)=-\sup_{F\in \mathcal{D}_n(\overline{\mathbf{F}})}q_{1-\alpha}^{-}(F)$, where $\overline{\mathbf{F}}=(\overline{F}_1,\dots, \overline{F}_n)$. Note that $\overline{\mathbf{F}}\in \left(\M_{1-t, D}^-\right)^n \cup \left(\M_{1-t, I}^-\right)^n$. Hence, using \eqref{eq:tech1} and \eqref{eq:tech2}, we have
 $$q_{\alpha}^+(F_{\mathcal{D}_n(\mathbf{F})}^+)=-\sup_{F\in \mathcal{D}_n(\overline{\mathbf{F}})}q_{1-\alpha}^{-}(F)=-\sup_{F\in \mathcal{D}_n(\overline{\mathbf{F}})}q_{1-\alpha}^{+}(F)=-H_{\overline{\mathbf F}}^{-1}(\alpha),~\alpha\in (0,t].$$
 Note that $H_{\overline{\mathbf F}}(x)$ is strictly decreasing and continuous over $(-\infty, \sum_{i=1}^{n}q_1^{-1}(\overline{F}_i))$ and is equal to $0$ on $[\sum_{i=1}^{n}q_1^{-1}(\overline{F}_i),\infty)$. Hence, we have $F_{\mathcal{D}_n(\mathbf{F})}^+(x)=H_{\overline{\mathbf F}}(-x),~x\leq -H_{\overline{\mathbf F}}^{-1}(t)$. This implies
$t\wedge F_{\mathcal{D}_n(\mathbf{F})}^+(x)=t\wedge H_{\overline{\mathbf F}}$. We can similarly show that the conclusion holds for $t=1$.

Clearly, $t\vee\widehat{F}_{\mathcal{D}_n(\mathbf{F})}^-$ and $t\wedge F_{\mathcal{D}_n(\mathbf{F})}^+$ are continuous.
 Next, we show the attainability of $\widehat{F}_{\mathcal{D}_n(\mathbf{F})}^-$ and $F_{\mathcal{D}_n(\mathbf{F})}^+$, which directly implies the attainability of $t\vee\widehat{F}_{\mathcal{D}_n(\mathbf{F})}^-$ and $t\wedge F_{\mathcal{D}_n(\mathbf{F})}^+$.  For any $x\in\R$, there exists a sequence of $F_k\in \mathcal{D}_n(\mathbf{F}), k\geq 1$ such that $\lim_{k\to\infty}F_k(x-)=\widehat{F}_{\mathcal{D}_n(\mathbf{F})}^-(x)$. For each $F_k$, there exists a copula $C_k$ such that the joint distribution of $(X_1,\dots, X_n)$ is $C_k(F_1,\dots, F_n)$ and $X_1+\dots+X_n\sim F_k$. It follows from Theorem 25.10 of \cite{B95} that we can find a subsequence $\{C_{k_m}, m\geq 1\}$ and a copula $C$ such that $C_{k_m}\to C$ weakly as $m\to\infty$, which implies $F_{k_m}\to F$ weakly as $m\to\infty$, where $(X_1,\dots, X_n)\sim C(F_1,\dots, F_n)$ and $X_1+\dots+X_n\sim F$.  Hence, we have $F\in\mathcal{D}_n(\mathbf{F})$. Note that there exists a sequence of $y_l\uparrow x$ such that $F$ is continuous at all $y_l$. Hence, we have $F(y_l)=\lim_{m\to\infty}F_{k_m}(y_l)\leq \liminf_{m\to\infty}F_{k_m}(x-)=\widehat{F}_{\mathcal{D}_n(\mathbf{F})}^-(x)$. Letting $l\to\infty$, it follows that $F(x-)\leq \widehat{F}_{\mathcal{D}_n(\mathbf{F})}^-(x)$, implying the attainability of $\widehat{F}_{\mathcal{D}_n(\mathbf{F})}^-$. We can similarly show the attainability of $\widehat{F}_{\mathcal{D}_n(\mathbf{F})}^+$.
\end{proof}
It is worthwhile to mention that the expressions of the {extremal distributions} in Proposition \ref{prop:aggregationF} are essentially the dual bounds in Theorem 4.17 of \cite{R13}.
In light of  Propositions \ref{prop:new} and   \ref{prop:aggregationF}, we obtain the following result.
\begin{theorem}\label{thm:4}
Suppose $F_1^{-1},\cdots, F_n^{-1}$ are continuous over $(0,1)$.
For  $\mathbf F\in \left(\M_{t,D}^+\right)^n \cup \left(\M_{t,I}^+\right)^n$ with $0\leq t<\alpha_{\Lambda}^-$, 
  let  $t\vee F_{\mathcal{D}_n(\mathbf{F})}^-$  be given by \eqref{eq:aggregation}. Then we have
  \begin{enumerate}[(i)]
 \item $\sup_{F\in \mathcal{D}_n(\mathbf{F})}\tilde{q}_{\Lambda}^-(F)=\tilde{q}_{\Lambda}^-(t\vee F_{\mathcal{D}_n(\mathbf{F})}^-)$;
\item  If $\Lambda$ is decreasing, then (i) remains true by replacing $\tilde{q}_\Lambda^-$ by $q_\Lambda^-$; Moreover, $\sup_{F\in\mathcal{D}_n(\mathbf{F})}\tilde{q}_{\Lambda}^+(F)=\sup_{F\in\mathcal{D}_n(\mathbf{F})}q_{\Lambda}^+(F)=\tilde{q}_{\Lambda}^+(t\vee F_{\mathcal{D}_n(\mathbf{F})}^-)$.
 \end{enumerate}
  For $\mathbf F\in \left(\M_{t,D}^-\right)^n \cup \left(\M_{t,I}^-\right)^n$ with $\alpha_{\Lambda}^+<t\leq 1$, let $t\wedge F_{\mathcal{D}_n(\mathbf{F})}^+$ be given by \eqref{eq:aggregationG}. Then we have
  \begin{enumerate}[(i)]
 \item $\inf_{F\in \mathcal{D}_n(\mathbf{F})}q_{\Lambda}^+(F)=q_{\Lambda}^+(t\wedge F_{\mathcal{D}_n(\mathbf{F})}^+)$;  $\inf_{F\in\mathcal{D}_n(\mathbf{F})}q_{\Lambda}^-(F)=q_{\Lambda}^-(t\wedge F_{\mathcal{D}_n(\mathbf{F})}^+)$;
\item  If $\Lambda$ is decreasing, then (i) remains true by replacing $q_\Lambda^+$  by $\tilde{q}_\Lambda^+$ and  $q_\Lambda^-$ by $\tilde{q}_\Lambda^-$.
 \end{enumerate}
 
 \end{theorem}

 Note that in (i) of Theorem \ref{thm:4}, the marginals have monotone densities beyond their corresponding $t$-quantiles for $0\leq t<\alpha_{\Lambda}^-$. If $\alpha_{\Lambda}^-$ is chosen to be sufficiently large such as $\alpha_{\Lambda}^-=0.9$ or $0.95$, then we can choose  $t$ to be very close to $1$ such as $t=0.89$ or $0.94$. In this case, (i) of Theorem \ref{thm:4} is valid for many commonly used marginal distributions in finance and insurance such as normal, lognormal,  $t$, exponential, gamma, and Pareto distributions. 
 
In what follows, we consider the portfolio selection problem, i.e., to choose a optimal portfolio position $\boldsymbol{w} \in \Delta_n$ such that
$$
R_\rho(\boldsymbol{w})=\rho\left(\sum_{i=1}^n w_i X_i\right)
$$ is minimized,  where $\rho$ is a risk measure and  $\left(X_1, \ldots, X_n\right)$ represents the negative returns of $n$ different assets. Here we suppose that $X_1,\dots, X_n$ have the  identical distribution $F$ as we aim to check whether diversification can reduce the risk.

Recall that a doubly stochastic matrix is a square matrix with nonnegative entries and the sum of each column/row is equal to 1. Let $\mathcal{Q}_n$ denote the set of all $n\times n$ doubly stochastic matrices.   For two portfolio positions $\boldsymbol{w}, \boldsymbol{\gamma}\in\Delta_n$, we can say that $\boldsymbol{\gamma}$ is more diversified than $\boldsymbol{w}$, denoted by $\boldsymbol{\gamma} \prec \boldsymbol{w}$, if $\boldsymbol{\gamma}={ A} \boldsymbol{w}$ for some ${ A} \in \mathcal{Q}_n$. Note that this binary relationship is also called the \emph{majorization order}; see e.g., \cite{M11}.
Here, we suppose the marginal distribution $X_i\sim F$ is known but the dependence structure is completely unknown. Hence, we consider the worst-case scenario. That is to find the optimal portfolio position $\boldsymbol{w}\in\Delta_n$ such that
$$
\overline{R}_\rho(\boldsymbol{w})=\sup \left\{\rho\left(\sum_{i=1}^n w_i Y_i\right): Y_1, \ldots, Y_n \sim F\right\}.$$
is minimized. Note that in the above setup, we only minimize the risk but do not consider the return of the portfolio. That is because the expected return of the portfolio is a constant and any constraint on the expected return makes no real sense.

\begin{proposition}\label{prop:diverse}
 Suppose that  $\rho=\tilde{q}_{\Lambda}^+$ with decreasing $\Lambda$ or $\rho=\tilde{q}_{\Lambda}^-$. Moreover, suppose that  $\left(X_1, \ldots, X_n\right)$ has an identical marginal distribution $F\in \M_{t,D}^+ \cup \M_{t,I}^+$ with $0\leq t<\alpha_{\Lambda}^-$ and $F^{-1}$ is continuous over $(0,1)$.  If $\boldsymbol{\gamma} \prec \boldsymbol{w}$ , then $\overline{R}_\rho(\boldsymbol \gamma) \geqslant \overline{R}_\rho(\boldsymbol{w})$. 
\end{proposition}
\begin{proof}
 Write $\boldsymbol{\gamma}=\left(\gamma_1, \ldots, \gamma_n\right)$ and $\boldsymbol{w}=\left(w_1, \ldots, w_n\right)$. Take $X \sim F$, and let $\mathbf{F}$ and $\mathbf{G}$ be the tuples of marginal distributions of $\left(\gamma_1 X, \ldots, \gamma_n X\right)$ and $\left(w_1 X, \ldots, w_n X\right)$, respectively. It follows that
\begin{equation*}\label{eq:bar R}
\overline{R}_\rho(\boldsymbol{\gamma})=\overline{\rho}(\mathbf{F}) \quad \text { and } \quad \overline{R}_\rho(\boldsymbol{w})=\overline{\rho}(\mathbf{G}).
\end{equation*}
 Using $\boldsymbol{\gamma} \prec \boldsymbol{w}$, there exists ${ A}\in \mathcal{Q}_n$ such that $(F_1^{-1},\dots, F_n^{-1})={ A}(G_1^{-1},\dots, G_n^{-1})$, denoted by $\mathbf{F}=A \otimes \mathbf{G}$.
By a slight extension of Theorem 3 of \cite{CLLW22}, for  $\alpha \in[t,1)$, we have $$\sup_{G\in \mathcal D_n(\mathbf{G})}q_{\alpha}^-(G)\leq \sup_{G\in \mathcal D_n({A} \otimes \mathbf{G})}q_{\alpha}^-(G)=\sup_{G\in \mathcal D_n(\mathbf{F})}q_{\alpha}^-(G).$$  Hence, it follows from (i) of Corollary \ref{coro:1} that $q_{\alpha}^-(F_{\mathcal D_n(\mathbf G)}^-)\leq q_{\alpha}^-(F_{\mathcal D_n(\mathbf F)}^-)$ for  $\alpha\in [t,1)$. By Proposition \ref{prop:aggregationF}, we have  $t\vee F_{\mathcal D_n(\mathbf F)}^-$ and $t\vee F_{\mathcal D_n(\mathbf G)}^-$ are continuous over $\R$. Hence, we have $t\vee F_{\mathcal D_n(\mathbf F)}^-(x)\leq t\vee F_{\mathcal D_n(\mathbf G)}^-(x)$ for all $x\in\R$.  By the monotonicity of $\rho$, we have $\rho(t\vee F_{\mathcal D_n(\mathbf F)}^-)\geq \rho(t\vee F_{\mathcal D_n(\mathbf G)}^-)$. Moreover, by Theorem \ref{thm:4},  we have $\overline{R}_\rho(\boldsymbol{\gamma})=\rho(t\vee F_{\mathcal D_n(\mathbf F)}^-)$ and $\overline{R}_\rho(\boldsymbol{w})=\rho(t\vee F_{\mathcal D_n(\mathbf G)}^-)$.  Hence, we have $ \overline{R}_\rho(\boldsymbol{w})\leq \overline{R}_\rho(\boldsymbol {\gamma}).$
\end{proof}

Note that Proposition \ref{prop:diverse} shows that for the assets with same marginal distributions $F\in \M_{t,D}^+ \cup \M_{t,I}^+$ with $0\leq t<\alpha_{\Lambda}^-$,  more diversified portfolio has higher risk under dependence uncertainty if  $\rho=\tilde{q}_{\Lambda}^+$ with decreasing $\Lambda$ or $\rho=\tilde{q}_{\Lambda}^-$ is applied. This is not surprising as some similar conclusions are shown in Proposition 8 of \cite{CLLW22} for quantiles under the assumption that the marginal distribution $F$ has a monotone density. Note that if $\Lambda$ is a constant, $\tilde{q}_{\Lambda}^-$ boils down to quantiles, and if $t=0$, then $F$ has a monotone density. Therefore, the conclusion in Proposition \ref{prop:diverse} can be viewed as an extension of Proposition 8 of \cite{CLLW22}.
Moreover, Proposition \ref{prop:diverse} also implies that under those assumptions,  the optimal portfolio positions are the ones with only a single asset.

 \section{Conclusion}\label{sec:5}

This paper summarizes some properties of $\Lambda$-quantiles and explores distributionally robust models of $\Lambda$-quantiles. We obtain robust $\Lambda$-quantiles on general uncertainty sets under some conditions, showing that obtaining the robust $\Lambda$-quantiles relies on finding the {extremal distributions} over the same uncertainty sets. This finding significantly simplifies  the problem,  enabling us to utilize many existing results on the robust quantiles from the literature.  We provide closed-form solutions for the uncertainty sets characterized by three different constraints: (i) moment constraints; (ii) probability constraints via the Wasserstein distance; and (iii) marginal constraints in risk aggregation.  These results are applied to optimal portfolio selection under model uncertainty, and  can also be extended to other problems such as optimal reinsurance, which will be discussed in the future.

\subsection*{Acknowledgments}
The authors thank the editor, an associate editor, and two anonymous referees for their helpful comments.
 The  research of  Xia Han  is supported by the National Natural Science Foundation of China (Grant Nos. 12301604, 12371471, and 12471449). Peng Liu is the corresponding author.

\subsection*{Data availability statement}

N/A. This is a theoretical paper and it does not make use of real data.

\subsection*{Conflict of interest disclosure}

None of the authors have a conflict of interest to disclose.

\appendix

\label{Appendix:proofs}

\end{document}